\documentclass[12pt]{article}
\usepackage{e-jc}
\usepackage{amsfonts,amsmath,amssymb,amsthm,latexsym,bm,mathtools,textcomp}
\usepackage{dsfont,frcursive}
\usepackage[all]{xy}
\usepackage[normal]{caption}

\usepackage[colorlinks=true,citecolor=black,linkcolor=black,urlcolor=blue]{hyperref}

\newcommand{\arxiv}[1]{\href{http://arxiv.org/abs/#1}{\texttt{arXiv:#1}}}
\usepackage[title,titletoc]{appendix}

\newcommand{\cA}{{\mathcal A}}
\newcommand{\cB}{{\mathcal B}}

\newcommand{\cM}{{\mathcal M}}
\newcommand{\cC}{{\mathcal C}}
\newcommand{\cQ}{{\mathcal Q}}
\newcommand{\cR}{{\mathcal R}}
\newcommand{\cS}{{\mathcal S}}

\newcommand{\cD}{{\mathcal D}}
\newcommand{\cG}{{\mathcal G}}
\newcommand{\cH}{{\mathcal H}}
\newcommand{\cI}{{\mathcal I}}
\newcommand{\cL}{{\mathcal L}}
\newcommand{\cT}{{\mathcal T}}
\newcommand{\cU}{{\mathcal U}}

\newcommand{\cUT}{\mathcal{UT}}
\newcommand{\cNT}{\mathcal{NT}}
\newcommand{\Hom}{\mathcal{H}om}

\def\fN{\mathfrak{N}}
\def\fF{\mathfrak{F}}

\def\rHW{\mathrm{HW}}

\def\rEnd{\mathrm{End}}
\def\rAut{\mathrm{Aut}}
\def\rOut{\mathrm{Out}}
\def\rInn{\mathrm{Inn}}
\def\rGL{\mathrm{GL}}
\def\rRFM{\mathrm{RFM}}

\def\rmd{\mathrm{d}}
\def\rr{\mathrm{r}}
\def\rs{\mathrm{s}}

\newcommand{\N}{{\mathbb N}}
\newcommand{\Z}{{\mathbb Z}}
\newcommand{\Q}{{\mathbb Q}}
\newcommand{\R}{{\mathbb R}}
\newcommand{\C}{{\mathbb C}}
\newcommand{\K}{{\mathbb K}}

\def\ra{\rightarrow}
\def\lra{\longrightarrow}
\def\map{\mapsto}

\newcommand{\laurent}[2]{#1 (\! ( #2 )\! )}

\newcommand{\puiseux}[2]{#1 \{\! \{ #2 \}\! \}}

\theoremstyle{plain}
\newtheorem{theorem}{Theorem}
\newtheorem{lem}[theorem]{Lemma}
\newtheorem{cor}[theorem]{Corollary}
\newtheorem{prop}[theorem]{Proposition}

\theoremstyle{definition}
\newtheorem{defi}[theorem]{Definition}

\newtheorem{comm}{Comment}
\theoremstyle{remark}
\newtheorem{rem}[theorem]{Remark}

\newtheorem{ex}[theorem]{Example}
\newtheorem{nota}[theorem]{Notation}
\newcommand{\bi}{\begin{itemize}}
\newcommand{\ei}{\end{itemize}}
\newcommand{\bd}{\begin{description}}
\newcommand{\ed}{\end{description}}
\newcommand{\be}{\begin{enumerate}}
\newcommand{\ee}{\end{enumerate}}

\newcommand{\dis}{\displaystyle}
\newcommand{\Union}[1]{\dis{\bigcup_{#1}}}

\newcommand{\Lim}[1]{\dis{\lim_{#1}}}
\newcommand{\Span}{\textrm{Span}}
\def\bc{\begin{center}}
\def\ec{\end{center}}
\def\Int{\dis \int}
\def\Frac{\dis \frac}
\def\Sum{\dis \sum}

\def\Lim{\dis \lim}

\newcommand{\tif}{\text{ if }}
\newcommand{\tand}{\text{ and }}

\newcommand{\stirc}[2]{\genfrac{[}{]}{0pt}{}{#1}{#2}}
\newcommand{\stirp}[2]{\genfrac{\{}{\}}{0pt}{}{#1}{#2}}
\newcommand{\bin}[2]{\genfrac{(}{)}{0pt}{}{#1}{#2}}

\def \vthx{\vartheta_x}
\def \partx{\partial_x}

\def\ov{\overline}
\def\und{\underline} 
\def\cirast{\circledast}

\def\no{\noindent}
\def\l{\left}
\def\r{\right}
\def\b{\big}

\def\bb{\bigskip} 
\def\m{\medskip}

\title{\bf Overview of the Heisenberg--Weyl Algebra \\ and Subsets of Riordan Subgroups}

\author{Silvia Goodenough \qquad Christian Lavault\\
\small LIPN, CNRS (UMR 7030)\\[-0.8ex]
\small Universit\'e Paris 13, Sorbonne Paris Cit\'e, France\\
\small\tt \{Silvia.Goodenough,lavault\}@lipn.univ-paris13.fr
}

\date{\dateline{May 17, 2015}{Oct 06, 2015}}

\begin{document}
\maketitle

\begin{abstract}
In a first part, we are concerned with the relationships between polynomials in the two generators of the algebra of Heisenberg--Weyl, its Bargmann--Fock representation with differential operators and the associated one-parameter group. 
Upon this basis, the paper is then devoted to the groups of Riordan matrices associated to the related transformations of matrices (i.e. substitutions with prefunctions). Thereby, various properties are studied arising in Riordan arrays, in the Riordan group and, more specifically, in the ``striped'' Riordan subgroups; further, a striped quasigroup and a semigroup are also examined. A few applications to combinatorial structures are also briefly addressed in the Appendix.

\bb\no \textbf{Keywords:} combinatorial algebra; Combinatorial models of creation--annihilation; Heisenberg--Weyl algebra; One parameter group; Riordan arrays and groups, Prefunction striped matrix; Striped Riordan subgroups; Operations for striped quasigroups and semigroups. 
\end{abstract}

\vskip 1cm
\begin{small}
\setcounter{tocdepth}{2}
\hypersetup{hidelinks}
\tableofcontents
\end{small}

\vskip 1.5cm
\pagenumbering{arabic}

\section{Introduction}
Quantum physics has revealed many interesting formal properties associated to an associative and unitary algebra of operators~$a$ and~$a^\dagger$ meeting the partial commutation relation~$aa^\dagger - a^\dagger a = 1$. The approach of considering an algebraic normal-ordering problem and have it transformed into a combinatorial (enumeration) problem is the source of a number of recent research works on the borders of quantum physics and algebraic or analytic combinatorics.

This work is mainly motivated by the algebraic introductory survey of Duchamp, Penson and Tollu~\cite{DuPT11}, by the article of Blaziak and Flajolet~\cite{BlFl11} and Blasiak, Dattoli, Duchamp, Penson and Solomon~\cite{Blaziak05,BHPS04,BlPS03a,DOTV97,DPSHB04} as well as by the whole bunch of recent papers on Riordan arrays and 
the Riordan group~\cite{Barry09,DaSW12,Hennessy11,SGWW91}, which extensively investigate the topic. Blaziak and Flajolet's paper provides a notably comprehensive and insighful synthetic presentation, especially with respect to the correspondence with combinatorial objects and structures, such as models that involve special numbers such as generalized Stirling numbers, set partitions, permutations, increasing trees, as well as weighted Dyck and Motzkin lattice paths, rook placement, extensions to $q$-analogues, multivariate frameworks, urn models, etc.

In the first four sections, the paper is concerned with the relationships between polynomials in the two generators of the algebra of Heisenberg--Weyl, the Bargmann--Fock representation of operators of creation--annihilation and the  transformations arising from the one-parameter group as Riordan matrices which are substitutions with prefunctions. Upon this basis, the goal of the second part of the paper is to adapt the former parts to the theory of Riordan arrays. The next four sections are thus devoted to the associated groups of Riordan matrices and, thereby, collect the various relations and properties arising in Riordan arrays and in the Riordan group. More specifically, ``striped'' Riordan subgroups, a quasigroup and a semigroup for two appropriate operations are respectively defined and studied in Section~\ref{extlaw}. A few  applications to combinatorial structures are also briefly mentioned in Appendix~\ref{app:combistruct}.

\section{The algebra of Heisenberg--Weyl} \label{hwa}
As customary in the associative and unitary Heisenberg--Weyl algebra of operators, the Lie bracket is the relation
\begin{equation} \label{eq:commrel}
[a, a^+] := aa^+ - a^+ a = 1,
\end{equation}
where~$a^+$ stands for the usual~$a^\dagger$, $1$ is the identity operator of the algebra and~$a a^+ - a^+ a$ is the commutator. This partial commutation relation is referred to as the \emph{bosonic commutation rule} or the \emph{creation--annihilation condition}. As a matter of fact, it is satisfied in quantum physics by the creation and annihilation operators $a$ and~$a^+$, which are adjoint to each other and serve to decrease or increase the number or the energy level of bosons by~1.
 
\begin{defi} \emph{(Heisenberg--Weyl algebra.)} \label{def:hwa}
From an abstract algebraic standpoint, one formally considers the algebra of Heisenberg--Weyl as the quotient
\begin{equation}
\rHW_{\C} = \C\langle A,B\rangle/\cI_{\rHW},
\end{equation}
where~$\C\langle A,B\rangle$ is the free associative algebra over~$\C$ of the free monoid~$\{A,B\}^{*}$ with identity~$1$. Namely, $\C\langle A,B\rangle$ stands for the algebra of polynomials in noncommuting indeterminates~$A$, $B$ and~$\cI_{\rHW}$ is the two-sided ideal generated by the polynomial~$AB - BA -1$. 
\end{defi}
Def.~\ref{def:hwa} is given by the mapping~$\mathfrak{s} : \C\langle A,B\rangle \ra \rHW_{\C}$ defined 
as~$a = \mathfrak{s}(A)$ and~$a^+ = \mathfrak{s}(B)$ (see~\cite{Blaziak05,BlPS03a,DuPT11}).

Since~$\rHW_{\C}$ is generated by~$a^+$ and~$a$, any element~$\omega\in \rHW_{\C}$ writes as a linear combination of finite products of such generators in the form
\begin{equation} \label{eq:lincombprodop}
\omega = \sum_{\bm{\rr},\bm{\rs}} \alpha_{\bm{\rr},\bm{\rs}} \l(a^{+}\r)^{r_1}a^{s_1}\cdots %
\l(a^{+}\r)^{r_j}a^{s_j},
\end{equation}
where~$\bm{\rr} = (r_1,r_2,\ldots,r_j)$ and~$\bm{\rs} = (s_1,s_2,\ldots,s_j)$ are multi-indices 
of non-negative integers (denoted by~$\Z_{\ge 0}$ or~$\N$) with the convention~$a^0 = a^{+^{0}} = 1$. 

Observe that the representation given by formula~\eqref{eq:lincombprodop} is ambiguous in so far 
as the rewriting rule of the commutation relation in eq.~\eqref{eq:commrel} allows different representations of a same element, e.g., $aa^+$ or equally~$a^+a + 1$. To remedy this situation, a preferred order of the generators is fixed by conventionally choosing the \emph{normally ordered} form in which all annihilators stand to the right 
of creators (see Wick's Theorem, e.g. in~\cite{BlFl11}).

\begin{defi} \label{defnfno} \emph{(Normal form, normal order.)} The commutation relation~\eqref{eq:commrel} may be regarded as a directed rewriting rule~$aa^{+}\, \lra \, a^{+}a+1$ (the normalization), which makes a systematic use of the reduction of~$aa^+ - a^+a - 1$ to~$0$. Any general expression~$\fF(a^{+},a)$ 
in~$\C\langle a^+,a\rangle$ is thus completely reduced to a \emph{unique} and \emph{equivalent} normal 
form~$\fN\b(\fF(a^{+},a)\b)\equiv \fF(a^{+},a)$, such that, in each monomial, all the occurrences of~$a^{+}$ precede all the occurrences of~$a$. 
\end{defi}
Every element of~$\rHW_{\C}$ is written~$\sum_{i,j} \beta_{i,j} b_{i,j}$ with~$b_{i,j} = \l(a^{+}\r)^i a^{j}$ in its normal form. By using~\cite[Chap.~2]{Bourbaki03}, one can show that~$\l(b_{i,j}\r)_{i,j\in \N}$ is a natural linear basis of~$\rHW_{\C}$ (the basis of normal forms).

\begin{lem} \label{strconst} From the normalization of an element~$\l((a^{+})^k a^\ell\r) \l((a^{+})^r a^s\r)$  ($k,\, \ell,\, r,\, s\in \N$) in~$\rHW_{\C}$, the \emph{structure constant} of $b_{p,q}$ is the 
coefficient~$c_{(s,\ell),(r,k)}^{(p,q)}$ ($p,\,q\in \N$) in the sum~$b_{s,\ell} b_{r,k} = %
\Sum_{p,q} c_{(s,\ell),(r,k)}^{(p,q)}\, b_{p,q}$, and
\begin{equation} \label{eq:sc}
(a^{+})^k a^\ell (a^{+})^r a^s = \sum_{i=0}^{\min(\ell,r)} i!\, \bin{\ell}{i}\bin{r}{i}\, (a^{+})^{k+r-i} a^{\ell+s-i}.
\end{equation}
This gives rise to a closed expression of the constant structure
\[
c_{(k,r)(\ell,s)}^{(k+r-i,\ell+s-i)} = i!\, \bin{\ell}{i} \bin{r}{i}\ \ \text{with}\ %
i=0,\ldots, \min(\ell,r).\]
\end{lem}

\begin{rem} \label{rem:nof}
Formula~\eqref{eq:sc} can be obtained either from an algebraic approach (similar to Wick's Theorem or Rook numbers)~\cite{DuPT11}, or (without loss of generality) by simply setting~$k = s = 0$ in~\eqref{eq:sc} and using Leibniz's rule on the appropriate operators~$\l(D^\ell[f] X^r[f]\r)\equiv D^\ell\l(X^r\r)[f]$ introduced further 
in~\S\ref{bfr}. (The beginning of this section consists of fundamental notions available in several articles~\cite{Blaziak05,DPSPBH10,BlFl11,BHPS04,DuPT11}, books~\cite[Chap.~2]{Hall04}, \cite[Chap.~6]{Vinberg03}, among (many) others; see also~\S\ref{nof} and Appendix~\ref{app:hla}.)
\end{rem}

\subsection{The Lie bracket in $\rHW_{\le 1}\!\setminus \! \rHW_0$} \label{Lieb}
$\rHW_{\le 1}$ denotes the subalgebra of elements in~$\rHW_{\C}$ for the usual product, which consist in polynomial operators of degree at most one in the variable~$a$, i.e. with only one or zero annihilation. By Def.~\ref{def:hwa}, such polynomials can be regarded as words in general form~$\omega = (a^{+})^{k} a^\delta %
(a^{+})^{\ell}$, where~$k,\, \ell$ are non-negative integers and the degree~$\delta$ of operator~$a$ equals 0 or 1. By the commutation relation~\eqref{eq:commrel}, the Lie bracket is a binary operation for~$\rHW_{\C}$ induced for the subalgebra~$\rHW_{\le 1}$.\par
It is assumed henceforth that~$\delta = 1$, since the specific case of degree zero corresponds to words with no annihilation, i.e. the set~$\rHW_0$ of scalar elements of~$\rHW_{\C}$. So, any two words~$\omega_1,\, \omega_2\in \rHW_{\le 1}\!\setminus \!\rHW_0$ are rewritten in the form~$\omega_1 = a^{+^{(k-r+1)}} a a^{+^r}$\ 
and~$\omega_2 = a^{+^{(\ell-s+1)}} a {a^+}^s$\ 
with~$k,\, \ell,\, r,\, s\in \N$. By Def.~\ref{defnfno}, this gives rise to
\begin{equation} \label{eq:liebra}
\omega_1 = a^{+^{k+1}} a + ra^{+^k}\ \quad \text{and}\ \quad \omega_2 = a^{+^{\ell+1}} a + sa^{+^\ell}\ %
\quad (k,\, \ell,\, r,\, s\in \N).
\end{equation}
The Lie bracket~$[\omega_1,\omega_2] = \b[a^{+^{k+1}}a + ra^{+^{k}} , a^{+^{\ell+1}}a + sa^{+^{\ell}}\b]$ is then
\begin{flalign*}
[\omega_1 , \omega_2] &\!=\! \l({a^+}^{k+1}a\r) \l({a^+}^{\ell+1}a\r) %
+ s\l({a^+}^{k+1}a\r) \l({a^+}^{\ell}\r)+r\l({a^+}^{k}\r)\l({a^+}^{\ell+1}a\r)+rs\l({a^+}^{k}\r)\l({a^+}^{\ell}\r) \\
-& \l({a^+}^{\ell+1}a\r) \l({a^+}^{k+1}a\r) + r\l({a^+}^{\ell+1}a\r) \l({a^+}^{k}\r) %
+ s\l({a^+}^{\ell}\r) \l({a^+}^{k+1}a\r) + rs\l({a^+}^{\ell}\r) \l({a^+}^{k}\r).
\end{flalign*}
After simplification and normalization of the above products, such as 
$\delimiterfactor=700 \l({a^+}^{k+1}a\r) \l({a^+}^{\ell+1}a\r)$, 
$\delimiterfactor=700 \l({a^+}^{k+1}a\r) \l({a^+}^{\ell}\r)$, etc.), the final expression of the Lie bracket makes all terms of degree two in the variable~$a$ disappear, and we obtain
\begin{equation} \label{eq:LieB}
[\omega_1, \omega_2] = (\ell-k){a^+}^{(k+\ell+1)}a + \b(s\ell - rk\b) a^{+^{(k+\ell)}}.
\end{equation}
\begin{rem} \label{rem:lb}
When~$|\omega_1| = |\omega_2|$ ($k = \ell$) in~$\rHW_{\le 1}\!\setminus \!\rHW_0$, the Lie bracket makes the vector part disappear and there remains the scalar part only in~eq.~\eqref{eq:LieB}.\par
As regards the algebra~$\rHW_{\C}$, it is unitary and associative and hence the Lie bracket in~$\rHW_{\C}$ fulfils the \emph{Jacobi identity}. As the Lie bracket also satisfies the axioms of \emph{bilinearity} and \emph{alternating 
on~$\rHW_{\C}$}, $\rHW_{\C}$ is a (graded) Lie algebra (see \S\ref{app:hla} in Appendix~\ref{app:hla}).
\end{rem}

\subsection{The Lie algebra $\rHW_{\le 1}\!\setminus \! \rHW_0$}
Consider an equivalence relation in~$\rHW_{\le 1}$ between words of equal lengths for which the annihilation~$a$ lays at different places.\par
In that case, for any~$n$, the words~${a^+}^{k_1}a{a^+}^{k_2}$ and~${a^+}^{\ell_1} a{a^+}^{\ell_2}$ 
with~$k_1 + k_2 = \ell_1 +\ell_2 = n$ belong to the same equivalence class in~$\rHW_{\le 1}\!\setminus \! \rHW_0$. Let~$\ov{{a^+}^na}$ be the representative of each class in that set and let $\fN\l(\rHW_{\le 1}\!\setminus \!\rHW_0\r)$ denote the set of the normalized elements of the 
subalgebra~$\rHW_{\le 1}\!\setminus \!\rHW_0$. Then, $\fN\l(\rHW_{\le 1}\!\setminus \!\rHW_0\r)\cong %
\rHW_{\le 1}\!\setminus \!\rHW_0$, since there exists a one-to-one correspondence between the sets of class representatives~$\b\{\ov{{a^+}^na}\b\}$ and~$\b\{{a^+}^na\b\}$. More precisely, given any operation 
for~$\b\{\ov{{a^+}^na}\b\}$, there exists also a unique operation for~$\b\{{a^+}^na\b\}$. Thus, these two structures are isomorphic.

\begin{theorem}
The Lie bracket on the set~$\rHW_{\le 1}\!\setminus \!\rHW_0$ defines a Lie algebra which is isomorphic to the set of vector fields on the real line.
\end{theorem}
\begin{proof}
The bases of~$\rHW_{\le 1}\!\setminus \!\rHW_0$ (as a vector space) are~$\b\{\langle \ov{{a^+}^na}\rangle\b\}$ 
and~$\l\{\langle{a^+}^na\rangle\r\}$ up to an equivalence, for all~$n \in \N$. From eq.~\eqref{eq:LieB}, the properties of the Lie bracket make~$\rHW_{\le 1}\!\setminus \!\rHW_0$ into a Lie algebra (i.e., bilinearity, alternating on~$\rHW_{\le 1}\!\setminus \!\rHW_0$ and the Jacobi identity).
\end{proof}

\section{Operators in the creation--annihilation model} \label{diffops}
The algebra of Heisenberg--Weyl comes with the natural representation~$\rHW_{\C} \ra \rEnd_\K(\Phi)$ ($\K = \R$ 
or~$\C$), where~$\Phi$ is some general class of \emph{smooth} functions. Thus, $\Phi$ may be the class~$\cC^\infty$ of infinitely differentiable functions over some suitable domain, the space of analytic functions, or the vector space~$\C[x]$ of polynomials in indeterminate~$x$, etc. For example, $\rEnd\b(\C[x]\b)$ is the algebra of endomorphisms 
of~$\C[x]$, which explains why it appears so often in many branches of mathematics and physics. 
The map~$\rHW_{\C} \ra \rEnd_\C(\Phi)$ is given on the generators of~$\rHW_{\C}$ through Bargmann--Fock representation of~$\rHW_{\C}$.

\subsection{The representation of Bargmann--Fock} \label{bfr}
A particular realization of the commutation relation is obtained by choosing some sufficiently general 
space~$\{f(x)\}$ of \emph{smooth functions} (typically~$\cC^\infty(0,1)$ or~$\C[x]$, see Comment~\ref{comm:shift}), on which two operators $X$ and~$D$ are defined as 
\[
X[f](x) = xf(x) \qquad \text{and} \qquad D[f](x) = \frac{\rmd}{\rmd x} f(x).\]
Then, the creation--annihilation principle is obviously satisfied by~$a = D$ and~$a^+ = X$. One recovers the Weyl relations~$[D,X]=1$ of abstract differential algebra~\cite[Ch.~1]{Ritt50}. The interest of such a differential model of creation--annihilation is that it is \emph{faithful}, meaning that any identity (without any additional assumption regarding the space of functions) is true in all generality under the commutation relation. This differential view will prove central to the following developments.

\begin{defi} \label{def:bfr}
The representation of Bargmann--Fock of~$\rHW_{\C}$ is given by the application 
$\phi\,:\, \rHW_{\C} \ra \Omega(\cC^\infty)$ with~$\phi(a^+) = x$ and~$\phi(a) = \frac{\rmd}{\rmd x}$, 
where~$\Omega(\cC^\infty) = \rEnd\l(\cC^\infty\b((0,1),\R\b)\r)$ is the set of such differential operators on a class of~$\cC^\infty$ functions over an appropriate domain.\par
The transformation of a term in~$\rHW_{\C}$ is given by a word whose letters are operators~$a^+$ and~$a$, where the multiplication is represented by~$x$ and the derivation by~$\frac{\rmd}{\rmd x}$, respectively. In this respect, 
$\phi(a^{+i} a^j) = x^i\l(\frac{\rmd}{\rmd x}\r)^j$ is a polynomial operator of excess~$e := |w|_{a^+} -  |w|_a = %
i - j\in \Z$ in~$\rEnd\b(\C[x]\b)$.
\end{defi}
For example, any normalized operator~$x^{k} \frac{\rmd}{\rmd x}$ in~$\Omega(\cC^\infty)$ on a smooth real 
function~$x^\ell f(x)$ defined on the open set~$(0,1)$ with~$k,\,\ell\in \N$, is written
\[
x^{k} \frac{\rmd}{\rmd x} \l[x^\ell f(x)\r] = x^{k} \l(\ell x^{\ell-1} f(x) + x^{\ell} %
\frac{\rmd}{\rmd x} f(x)\r) = \l(x^{(k+\ell)}\frac{\rmd}{\rmd x} + \ell x^{(k+\ell-1)}\r) f(x),\]
where~$x^{(k+\ell)}\frac{\rmd}{\rmd x}$ is the vector field part of the operator 
and~$v(x) := \ell x^{(k+\ell-1)}$ is its scalar field part on the real line. For~$\ell = 0$, $v\equiv 0$, 
the operator reduces to a pure tangent vector field.\par
From now on, we will denote~$\partx := \frac{\rmd}{\rmd x}$\ and~$\vartheta_x := x\partx$ or 
also~$\partx := D$\ and~$\vartheta_x := XD$, as the case may be.

There follows the computation of the Lie bracket of two differential operators in~$\rHW_{\C}$.

\subsection{The Lie Bracket of two differential operators} \label{liebrdop}
Let~$q_1(x)\partx \,+\, v_1(x)$ and~$q_2(x)\partx \,+\, v_2(x)$ be two differential operators (the sum of one vector field and one scalar field on the line). By making use of Poisson's braces notation~$\b\{q_1(x),q_2(x)\b\} = q_1(x)q_2'(x) - q_2(x)q_1'(x)$, the Lie bracket expresses as
\begin{flalign} \label{eq:poisson}
\b[q_1(x)\partx + v_1(x),q_2(x)\partx + v_2(x)\b]\!\! 
&=\!\! \b(q_1(x)q_2'(x) - q_2(x)q_1'(x)\b)\partx + q_1(x)v_2'(x) - q_2(x)v_1'(x)\nonumber\\
&={} q_1(x) \b(\l(q_2'(x)\partx + v_2'(x)\r)\b) - q_2(x) \l(q_1'(x)\partx + v_1'(x)\r)\nonumber\\ 
&={} \b\{q_1(x),q_2(x)\b\} \partx + q_1(x)v_2'(x) - q_2(x)v_1'(x).
\end{flalign}

\begin{ex} \label{ex1}
If there is no scalar part, $v_1(x) = v_2(x) \equiv 0$ and, by eq.~\eqref{eq:poisson},
\[
\b[q_1(x) \partx , q_2(x) \partx \b] = \b\{q_1(x),q_2(x)\b\} \partx .\]
Taking~$q_1(x) = x^{k+1}$ and~$q_2(x) = x^{\ell+1}$ yields~$\b[x^{k} \vthx , x^{\ell} \vthx \b] %
= (\ell-k)x^{\ell+k}\vthx $ with $k,\,\ell$ non-negative integers.\par
When there exists a scalar part ($v(x)\not \equiv 0$) and under the assumption that, either vector parts meet the condition~$v(x) = q(x)/x$\ ($x\ne 0$), or all operators have the form~$q(x)\vthx + tq(x)$ ($t\in \Q$, $t > 1$). Then, by~eq.~\eqref{eq:poisson} (Poisson's notation), the Lie bracket is written as
\begin{flalign*}
\b[q_1(x)\vthx + rq_1(x) , q_2(x)\vthx + sq_2(x)\b] %
&={} \b\{xq_1(x),xq_2(x)\b\} \partx + x \b(q_1(x)v_2(x) - q_2(x)v_1(x)\b)\\
={} x & \b\{q_1(x),q_2(x)\b\} \vthx + x \b(sq_1(x)q_2'(x) - rq_2(x)q_1'(x)\b).
\end{flalign*}
Specifically, when~$q_1(x) = x^k$ and~$q_2(x) = x^\ell$, where $k,\,\ell,\,r,\,s\in \Z_{>0}$, one gets
\begin{flalign*}
\b[x^{k}\vthx + rx^k , x^{\ell}\vthx + sx^\ell\b] &= (\ell-k) x^{k+\ell}\vthx + (s\ell - rk) x^{k+\ell}, \\
\b[x^{k}\vthx - rx^k , x^{\ell}\vthx - sx^\ell\b] &= (\ell-k) x^{k+\ell}\vthx - (s\ell + rk) x^{k+\ell}.
\end{flalign*} 
Eq.~\eqref{eq:LieB} is recovered in the first equality. Notice also that, in either case, if~$k = \ell$, the Lie bracket applies to two operators of equal degrees and therefore reduces to its scalar parts only. 
\end{ex}

\subsection{Normally ordered powers of strings in $\rHW_{\C}$} \label{nof}
As noticed in Section~\ref{hwa}, the normally ordered form of a general expression~$\fF(X,D)$ of~$\C\langle X,D\rangle$ is obtained by making use of the commutation relation~\eqref{eq:commrel} by moving all the annihilation operators~$a$ to the right. Whence a variety of properties provided by the normal ordering. All coefficients of the normal ordering of a boson string (or \emph{word})~$\omega\in \{X,D\}^{*}$ (more precisely the coefficients of the decomposition of~$\omega$ on the basis~$(X^i D^j)_{i,j\in \N}$) are positive integers, which suggests that such integers count combinatorial objects. For example, the Bell and Stirling numbers have a combinatorial origin, nevertheless one may consider them (and their generalizations) also as coefficients of the normal ordering~\cite[\S4.2--4.3]{BlFl11} and~\cite{Blaziak05,BlPS03a,DuPT11}. (Generalized $\omega$-Stirling numbers are reintroduced in this way in eq.~\eqref{eq:genstir2} below.)

Before the representations (or realizations) of the one-parameter group~$\l\{e^{\lambda \omega}\r\}_{\lambda \in \R}$ are defined formally, one must undertake the problem of ordering the powers of~$\omega\in \rHW_{\C}$ in normal 
form~$\fN\l(\omega^n\r) = \Sum_{n,i,j} \alpha(n,i,j) X^i D^j$. Though, in general, such is a three parameters problem, it can be reduced to two parameters for homogeneous operators by help of the gradation property introduced in Appendix~\ref{app:hla}. In this respect, any normalized polynomial operator~$\omega\in \rHW_{\C}$ can be expressed 
as~$\fN\l(\omega\r) = \Sum_{i-j=E} \alpha(i,j) X^iD^j$, where~$E\in \Z$ denotes the degree (or excess) of~$\omega$. 
\begin{equation} \label{eq:genstir2}
\fN\l(\omega^n\r) = 
\begin{dcases}
X^{nE} \sum_{k\ge 0} \stirp{n}{k}_{\omega} X^k D^k & \text{if}\ E\ge 0 \\ 
\l(\sum_{k\ge 0} \stirp{n}{k}_{\omega} X^k D^k\r) D^{n|E|} & \text{if}\ E < 0.
\end{dcases}
\end{equation}
This is the definition of the generalized~$\omega$-Stirling numbers of the second kind, as recently introduced and used for polynomials and generalized to homogeneous operators, e.g. in~\cite{Blaziak05,BlFl11,BlPS03a}.

\begin{ex} \label{ex:genstirling2}
\emph{(i)} For balanced polynomials, which are of the form~$\sum_k \stirp{n}{k}_{\omega} X^k D^k$ ($E = 0$ 
in~eq.~\eqref{eq:genstir2}), the~$\omega$-Stirling numbers admit the explicit form
\[
\stirp{n}{k}_{\omega} = \frac{1}{k!} \sum_{j=1}^k (-1)^{k-j} \bin{k}{j} h(j)^n\ \qquad \text{with}\ %
h(x)  := \sum_{k\ge 1} \alpha(k) x^{\und{k}}.\]
\emph{(ii)} The algebraic reduction of operators~$\l(X^r D^s\r)^n$ with~$r, s, n\in \N$ has the form
\begin{flalign} 
\l(X^r D^s\r)^n &= X^{nE} \sum_k \stirp{n}{k}_{r,s} X^k D^k\ \qquad \tif\ E\ge 0, \nonumber\\
\l(X^r D^s\r)^n &= \l(\sum_k \stirp{n}{k}_{r,s} X^k D^k\r) D^{n|E|}\ \;\tif\ E < 0. \label{eq:genstirling2}
\end{flalign} 
The above formula obviously holds in the case when~$E\ge 0$. The case~$E < 0$ also gives rise to similar coefficients (see Lemma~\ref{strconst}), as results from the ``duality argument'' shown in~\cite[Note~2, p.~11]{BlFl11}. $\stirp{n}{k}_{r,s}$ denotes an operator formulation of generalized Stirling numbers of the second kind, which appears, for example, in Comtet's book~\cite[Ex.~2, p.~220]{Comtet74}, 
in Katriel~\cite{Katriel74} and, recently, in~\cite{BlFl11,DuPT11,FlSe09}). When~$r = s = 1$, 
$\stirp{n}{k}_{1,1} := \stirp{n}{k}$ denotes the usual Stirling numbers of the second kind as defined after Stirling in 1730 and (re)explored later from a combinatorial viewpoint, e.g. by Carlitz~\cite{Carlitz32}, 
Comtet~\cite{Comtet74}, \cite[\S6.1]{GrKP89}, Riordan~\cite[\S6.6]{Riordan68}, etc.
\end{ex}

\begin{rem} \label{rem:unipotent}
The words~$\omega\in \rHW_{\le 1}\!\setminus \! \rHW_0$ (with one annihilation only) have the general form~$w = %
(a^+)^{n-p} a (a^+)^p$. If~$p = 0$, $\l(\stirp{n}{k}_{\omega}\r)$ is the matrix of a \emph{unipotent substitution}, while if~$p > 0$, $\l(\stirp{n}{k}_{\omega}\r)$ is the matrix of a \emph{unipotent substitutions with prefunctions} (see Lemma~\ref{lem:rfalg}). This will prove of prime importance in the integration of the one-parameter group in next Section~\ref{onepargroup}.
\end{rem}

\section{One-parameter groups} \label{onepargroup}
In general, a lot of meaningful operators can be associated to elements in~$\rHW_{\C}$. In particular, for any given polynomial~$\omega\in \C\langle X,D\rangle$, the one-parameter group~$\l\{e^{\lambda \omega}\r\}_{\lambda\in \R}$ with sufficienly small parameter~$\lambda$ is of particular interest in quantum physics and for related identities on combinatorial numbers. (Full details on one-parameter groups may be found e.g. in~\cite{DuPT11,DPSHB04,Gilmore08}.)

Now, introduce the exponential as an operator on~$\C\langle X,D\rangle[[z]]$. The notion of exponential operator is developed by Roman in~\cite[Chap.~2]{Roman05} and Dattoli \emph{et al.} in~\cite[Part~I]{DOTV97}, and rephrased in terms of exponential generating functions (EGF) by Blaziak and Flajolet in~\cite[\S2.3]{BlFl11}.
\begin{equation}
e^{z\omega} := \sum_{n\ge 0} \omega^n \frac{z^n}{n!}
\end{equation}
is a power series in~$z$ whose coefficients are in the polynomial ring~$\C\langle a,a^{+}\rangle$ with the normal form
\begin{equation}
\fN(e^{z\omega}) = \sum_{n\ge 0} \fN(\omega^n) \frac{z^n}{n!}.
\end{equation}
The sum on the right-hand side is no other than the EGF of the total number of specific combinatorial objects, for example the total weight of the corresponding \emph{diagrams} in the combinatorial context of~\cite[\S2.1]{BlFl11} and~\cite[p.~531]{FlSe09}.

\begin{comm} \label{comm:shift}
It is a well-known property that the exponential of a derivative plays the role of a \emph{shift (or translation) operator}. As an example, the operator~$e^{\lambda D}$ ($\lambda\in \R$, $D := \partial$) may be interpreted symbolically through its Taylor's expansion in~$\lambda$. 
\[
e^{\lambda D}[f](x) = \sum_{n\ge 0} \lambda^n/n! D^n [f(x)] = \sum_{n\ge 0} \lambda^n/n! f^{(n)}(x) = f(x + \lambda).\]
The action of the shift operator on the monomial~$x^n$ is evident by the binomial theorem, and thus on all polynomials~$f\in \C[x]$. This is also the case on all (formally convergent) series~$f\in \C[[x]]$, due to the purely symbolic nature of the calculation (see Appendix~\ref{app:fpsgf}). Taylor formula also applies to any complex (or analytic) polynomial, thus preserving the shift operator property. Recall that a function~$f$ is said to be of class~$\cC^\infty$, or smooth, if it has derivatives of all orders. Now, consider an open set~$D$ on the real line (e.g. $D = (0,1)$ in the present setting of Def.~\ref{def:bfr}) and a real valued function~$f$ defined on~$D$. Any real valued function~$f$ is said~$\cC^\omega(D)$ (or real analytic on~$D$) if it is smooth \emph{and} if it equals its Taylor series expansion around any point in its domain (thus $\cC^\omega\subset \cC^\infty$). Such is also the case of $\cC^\infty$ complex functions on~$D$, for any complex function which is differentiable (in the complex sense) in an open set is analytic, or holomorphic. (The notion of shift operator is comprehensively investigated by Rudin in~\cite[17.20--17.23]{Rudin82}.)\par
Hence, in the context of Section~\ref{onepargroup}, we must first consider operators of the form~$e^{\lambda D}$ (for a sufficiently small parameter~$\lambda\in \R$) on a typically general space of holomorphic functions on a suitable domain, such as~$(0,1)$. Besides, one has also to assume that~$|\lambda| < R$, where~$R$ is the radius of convergence of~$f$ in order to ensure at least the necessary conditions of formal displacements for series in~$z\C[[z]]$~\cite[Ch.~IV.4.3.]{Bourbaki06}. Fortunately, this is verified in the situation of such one-parameter groups.

Thereby, holomorphic functions on~$(0,1)$ appear a suitable candidate in being the actual maximal class of functions satisfying the shift operator requirements.
\end{comm}

\subsection{Integration of the one-parameter group} \label{intopg}
The normal ordering of elements in~$\rHW_{\le 1}\!\setminus \! \rHW_0$ (with one annihilation only), i.e. differential operators of the first order exactly, is shown (e.g. in~\cite{Bourbaki97,DPSHB04}) to be of a special type. More precisely, the integration of the one-parameter groups generated by such operators involves transformations going on the name of \emph{substitutions with prefunctions} in combinatorial physics, that is actually \emph{Riordan arrays}. In the following, we are concerned with monomials in the general form
\begin{equation} \label{eq:mongf}
q(x)\partx + v(x),
\end{equation}
the sum of one vector field and one scalar field on the line.

The object of this subsection is the integration of the one-parameter groups associated to~\eqref{eq:mongf}. Taking a geometric viewpoint, one uses the fact that any such one-parameter group is conjugate to the pure vector 
field~$q(x)\partx$ on the line (\emph{tangent vector field paradigm}). This will result further in the 
one-parameter group~$f\map U_\lambda[f]$ in the form
\begin{equation} \label{eq:basicdop} \delimiterfactor=1200 
U_\lambda [f](x) = e^{\lambda \l(q(x)\partx + v(x)\r)} [f](x) = g_\lambda(x) f\l(s_\lambda(x)\r),
\end{equation}
where~$g_\lambda$ and~$s_\lambda$\ are both analytic in a neighbourhood of the origin; and under the assumptions 
that~$q$ and~$v$ are at least continuous, $\lambda\in \R$ is sufficiently small and~$f$ is a function in an appropriate space (see Comment~\ref{comm:shift}). The calculations of $s_\lambda$ and~$g_\lambda$ are deferred until next~\S\ref{subst1} in Appendix~\ref{app:intopg} and what follows.

\subsection{Substitutions with prefunctions in $\rHW_{\le 1}\!\setminus \! \rHW_0$} \label{substpref}
Following~\cite{DPSHB04,DuPT11}, the integration of the one-parameter group~$e^{\lambda \l(q(x)\partx %
+ v(x)\r)} \l[f(x)\r]$ can be considered in the first place when the problem involves a pure vector field only (i.e. $v\equiv 0$).

\subsubsection{Evaluation of the substitution in a pure vector field} \label{subst1}
In this case, the transformation of~$f(x)\ne 0$ is given by a substitution factor~$s_\lambda$ only, which means that eq.~\eqref{eq:basicdop} admits the form
\begin{equation} \label{eq:vfdop} \delimiterfactor=1100
e^{\lambda q(x)\partx} \l[f(x)\r] = f\l(s_\lambda(x)\r).
\end{equation}
The computation of that substitution function is given in Appendix~\ref{app:intopg}. Now, whenever an expression takes the (general) form~$a^{+^{n}} a$ with integer~$n\ge 2$, the operator is~$e^{(\lambda x^n \partx )}$, which gives rise to the substitution function 
\begin{equation} \label{eq:vfsubs} 
s_\lambda(x) = \frac{x}{\sqrt[\uproot{2} n-1]{1 - (n-1) \lambda x^{n-1}}}, 
\end{equation}
absolutely convergent for~$|x| < 1\b/\b.\! \sqrt[n-1]{(n-1) \lambda}$. Formulas~\eqref{eq:vfdop}-\eqref{eq:vfsubs} translate into
\begin{equation} \label{eq:vfresult} \delimiterfactor=1100
e^{\lambda x^n \partx} \l[f(x)\r] = \delimiterfactor=700 %
f\l(\frac{x}{\sqrt[\uproot{2} n-1]{1 - (n-1) \lambda x^{n-1}}}\r). 
\end{equation}
For any~$n\in \Z_{>0}$, the related geometric transformation turns out to be the conjugate of a homography and a transformation~$\pi_n : x\map x^n$, whose reversal is~$\ov{\pi}_n : x\map x^{1/n}$, where~$\pi_n$ and~$\ov{\pi}_n$ are both~$\cC^\infty$ in an appropriate space depending on~$f$. Let~$h_n$ denote the homography relative to the substitution function~$s_\lambda(x) = \frac{x}{1-\lambda nx}$. Since~$h_n$ and~$\pi_{n}$ are conjugates with respect to the composition, $s_\lambda$ writes as
\begin{equation} \label{eq:sn2}
s_\lambda = \ov{\pi}_n \circ h_n \circ \pi_n.
\end{equation}
Actually, for any~$n > 0$ the substitution simplifies to~$s_\lambda(x) %
= \l(\frac{x^n}{1 - n \lambda x^n}\r)^{1/n} = \frac{x}{\l(1 - n\lambda  x^n\r)^{1/n}}$. 
Now, let~$F_n(x) = \frac{x^n}{(1 - n\lambda  x^n)}$, then~$s_\lambda(x) = F_n(x)^{1/n}$ and, setting~$y = x^n$, 
$F_n(x) = \frac{y}{1 - n\lambda y}$, which is the homography~$h_n$. Finally, for any~$n > 0$, 
$\frac{y}{1 - n\lambda y} = \l(\frac{y}{1 - \lambda y}\r)^{(n)}$ and~$h_n := h_1^{(n)}$, 
where the exponent denotes the~$n$th compositional power of current functions (in accordance with Notations~\ref{nota} in~\S\ref{stripedmat}).

By plugging~$h_n$ into eq.~\eqref{eq:sn2}, one can rewrite~$s_\lambda$ for any integer~$n\ge 1$ in the simpler form
\[
s_\lambda = \ov{\pi}_n \circ h_1^{(n)} \circ \pi_n.\]
The cases of~$n =$ 0, 1 and 2 are treated in the example of Appendix~\ref{app:intopg}, where the corresponding substitution functions~$s_\lambda$ evaluate respectively to a translation ($n = 0$), to an homothety ($n = 1$) and to a homography ($n = 2$).

\begin{ex} \label{subst2} (The Lie bracket and the substitution factor.)\ 
Given two positive integers~$k$ and~$\ell$, take~$q_1(x) = x^{k+1}$ and~$q_2(x) = x^{\ell+1}$. 
From the previous~\S\ref{subst1}, the substitution functions associated to~$q_1(x)$ and~$q_2(x)$ are, respectively, 
\[
s_{\lambda;k}(x) = \frac{x}{\l(1- k\lambda x^k\r)^{1/k}}\ \qquad \text{and}\ %
\qquad s_{\lambda;\ell}(x) = \frac{x}{\l(1 - \ell\lambda x^\ell \r)^{1/\ell}}.\]
Hence, the Lie bracket of the two operators 
is~$\b[x^{k}\vthx , x^{\ell}\vthx \b] = (\ell-k)x^{\ell+k} \vthx $ and~$s_{\lambda}(x)$, standing for the substitution function, admits the general form (according to the sign of~$\ell - k$) 
\begin{equation} \label{eq:slambdaLie} \delimiterfactor=1200 
s_{\lambda}(x) = \frac{x}{\sqrt[\uproot{2} k+\ell]{1 \pm |\ell-k|(k+\ell)\lambda x^{k+\ell}}}.
\end{equation}
This formula is at the basis of the three cases discussed in~\S\ref{Liebprefun}.
\end{ex}

\subsubsection{Substitutions with prefunction: the general case of a vector field} \label{subst3}
In the general case of a vector field with scalar part~$v\not \equiv 0$, the integration of the associated one-parameter group~$e^{\lambda \l(q(x)\partx + v(x)\r)}$ results in the basic form~\eqref{eq:basicdop},
\[ \delimiterfactor=1200 
U_\lambda [f](x) = e^{\lambda \l(q(x)\partx + v(x)\r)} \l[f(x)\r] = g_\lambda(x) f\l(s_\lambda(x)\r),\]
under the assumptions on~$q(x)$, $v(x)$, parameter~$\lambda$ and $f(x)$ already drawn in \S\ref{intopg} (from 
Comment~\ref{comm:shift} and Appendix~\ref{app:intopg}). A few transformations in eq.~\eqref{eq:vfdop} allow to integrate the one-parameter group~$e^{\lambda \l(q(x)\partx + v(x)\r)}$ for a general scalar field~$v(x)$.

Taking again a geometric viewpoint, we make use here of the fact that a general field of type~$q(x)\partx + v(x)$ is conjugate to the tangent vector field~$q(x)\partx$ on the line (with respect to composition). So, on the same assumptions as in Appendix~\ref{app:intopg} and~\S\ref{subst1}, 
let~$u(x) = \exp \l(\Int_{x_0}^x \frac{v(t)}{q(t)} \rmd t\r)$.\par
The function~$u(x)$ satisfies~$q(x)\partx + v(x) = u^{-1}(x) \b(q(x)\partx \b) u(x)$, which gives the conjugate to the tangent vector field~$q(x)\partx$ in the neighbourhood of the origin $\lambda = 0$ (see Comment~\ref{com:conjtrick}). Due to the fact that the exponential commutes with the conjugacy, we thus have
\begin{equation} \label{eq:conj2}
\delimiterfactor=1100 U_\lambda = e^{\lambda \l(q(x)\partx + v(x)\r)} = u^{-1}(x) e^{\lambda q(x)\partx} u(x),
\end{equation}
Then, by the calculations performed in~\S\ref{subst1} and in Appendix~\ref{app:intopg}, one obtains the integration of the general one-parameter group (at least locally) under the transformation~$f\map g_\lambda(f \circ s_\lambda)$: 
\begin{equation} \label{eq:onepargroup}
\delimiterfactor=1100 U_\lambda[f](x) = e^{\lambda \l(q(x)\partx + v(x)\r)} \l[f\r](x) %
= \frac{u\l(s_\lambda(x)\r)}{u(x)}\, f\l(s_\lambda(x)\r),
\end{equation}
where~$s_\lambda$ is the substitution factor with prefunction~$g_\lambda = (u\circ s_\lambda)/u$. 

\begin{comm} \label{com:conjtrick} The ``conjugacy trick'' and the ``tangent paradigm'' (so called in~\cite{DPSHB04}) which lay behind the result in eq.~\eqref{eq:conj2} may explain themselves as follows. 

Regarding vector fields as infinitesimal generators of one-parameter groups leads to conjugacy since, if~$U_\lambda$ is a one-parameter group of transformation, so is~$V U_\lambda V^{-1}$ ($V$ being a continuous invertible operator). In the context, we can formally consider~$(a^+)^{n-p} a (a^+)^p$ with $p > 0$ as conjugate to~$(a+)^n a$. 

\begin{ex} \label{ex:conjtrick} (\emph{Conjugacy trick})\ \ More generally, supposing all the terms well-defined, let~$u_2(x) := \exp \l(\int_{x_0}^x \frac{v(t)}{q(t)} \rmd t\r)$ and~$u_1(x) := q(x)/u_2(x)$. Then, since 
\[
u_1(x) u_2'(x) = u_1(x) u_2(x) v(x)/q(x) = v(x),\] 
one gets~$u_1(x)\partx u_2(x) = v(x)$. So, the operator~$q(x)\partx + v(x)$ reads as
\begin{equation} \label{eq:conjtrick}
u_1(x)u_2(x)\partx + u_1(x) u_2'(x) = u_1(x) \b(u_2'(x) + u_2(x)\partx\b) %
= u_2^{-1}(x)\b(u_1(x) u_2(x)\partx \b)u_2(x).
\end{equation}
Eq.~\eqref{eq:conjtrick} is conjugate to a vector field and integrates as a substitution with prefunction factor. Finally, by straightening the vector field on the line by the technique described in Remark~1, Appendix~\ref{app:intopg}, the required formula in~\eqref{eq:conj2} holds. (This method also amounts to use the ``ad'' operator (conjugacy) of derivation in the Lie algebra.)
\end{ex}
Now, the ``tangent paradigm'' works in the following manner: if the tangent vector is adjusted so as to coincide 
with~$x^{n-p}\partx x^p$, then we get the right one-parameter group. By virtue of the ``conjugacy trick'' and the ``tangent paradigm'', the integration of the one-parameter group yields 
\begin{flalign} \label{eq:final}
U_\lambda[f](x) &= e^{\lambda \omega}[f](x) =\delimiterfactor=700  \l(\frac{s_\lambda(x)}{x}\r) \delimiterfactor=900 %
f\b((s_\lambda(x)\b)\ \quad \text{(possibly locally),}\ \qquad \text{where}\nonumber\\
s_\lambda(x) &= \frac{x}{\sqrt[\uproot{2} n-1]{1 - (n-1) \lambda x^{n-1}}}\ \quad \tand\ \quad %
g_{\lambda}(x) = \frac{1}{\sqrt[\uproot{2} n-1]{1 - (n-1) \lambda x^{n-1}}}\,.
\end{flalign}
It can be checked that, if~$s_\lambda(x)$ is a substitution factor, in other words (at least locally) 
$s_{\lambda_1} \b(s_{\lambda_2}(x)\b) = s_{\lambda_1+\lambda_2}(x)$, such that~$s_\lambda(0) = 0$ for every~$\lambda$, then the transformations defined by~$U_\lambda[f](x) = \l(\frac{s_\lambda(x)}{x}\r) f\b((s_\lambda(x)\b)$ form a one-parameter (possibly local) group.
\end{comm}

\begin{rem} \label{rem:tgvect}
The transformations~$U_\lambda[f](x)$ are conducted near~$x = 0$ and the result must also stay in that neighbourhood. Thus, whenever the composition of two such transformations, say~$U_{\lambda_1}$ and~$U_{\lambda_2}$, is realized, the values~$\lambda_1$ and~$\lambda_2$ of parameter~$\lambda$ have to be chosen small enough to 
keep~$\l(U_{\lambda_1} \circ U_{\lambda_2}\r)[f](x)$ stay also close to zero. Then, the correctness of the computational procedure can also be stated \emph{a posteriori} by making use of a technique of tangent vector as follows (see~\cite{DuPT11}).\par
Check that, for small values~$\lambda_1$ and~$\lambda_2$ of the parameter~$\lambda$, 
$U_{\lambda_1} \circ U_{\lambda_2} = U_{\lambda_1+\lambda_2}$\ (local one-parameter group) and check 
that~$\frac{\rmd}{\rmd \lambda}\b|_{\lambda=0}\, U_\lambda[f](x) = \b(q(x)\partx + v(x)\b)[f](x) %
= \b(u^{-1}(x)\, q(x)\partx \,u(x)\b)[f](x)$ (tangent vector field at the origin). Then the transformations involved define substitutions with prefunctions factors: $f\, \map\, g_\lambda(f \circ s_\lambda)$.
\end{rem}

\subsubsection{The Lie bracket and the prefunction} \label{Liebprefun}
As examined in Ex.~\ref{ex1} in \S\ref{liebrdop}, given~$k,\, \ell$ and~$r,\, s$ in~$\Z_{>0}$, consider the words~$\omega_1 = a^{+^{(k+1)}} a a^{+^r}$\ and~$\omega_2 = a^{+^{(\ell+1)}} a a^{+^s}$ in~$\rHW_{\le 1}\!\setminus \!\rHW_0$, whose normal forms are~$\fN\l(\omega_1\r) = a^{+^{(k+1)}} + ra^{+^k}$\ and~$\fN\l(\omega_2\r) = a^{+^{(\ell+1)}} + sa^{+^\ell}$. They are represented here as~$x^{k}\vthx + rx^{k}$\ and~$x^{\ell}\vthx + sx^\ell$, respectively (the case 
of~$x^{k}\vthx - rx^{k}$ and~~$x^{\ell}\vthx - sx^\ell$, as given in Ex.~\ref{ex1}, is discussed in Remark~\ref{rem:genliebra}). The Lie Bracket and the prefunction follow:
\begin{flalign} \label{eq:glambda}
\b[x^{k} \vthx + rx^{k} , x^{\ell} \vthx + sx^\ell\b] & = (\ell-k) x^{k+\ell} \vthx + \b(s\ell - rk\b) x^{k+\ell}\ \qquad \tand\nonumber \\
g_\lambda(x) ={} & \delimiterfactor=1200 \frac{1}{\b(1 - (\ell-k)(k+\ell)\lambda x^{k+\ell}\b)^{\frac{s\ell-rk}{(\ell-k)(k+\ell)}}}\,.
\end{flalign}
To draw conclusions as to the characteristics of the prefunction~$g_\lambda(x)$ in eq.~\eqref{eq:glambda}, the problem at stake is to determine the sign of~$\frac{s\ell-rk}{(\ell-k)}$, according to the respective values of integers~$k, \ell, r, s\ge 1$.

First, without loss of generality, we can assume~$k\ne \ell$, otherwise $g_\lambda(x)\equiv 1$ (see~\S\ref{subst2}). Next, the case of~$\ell > k$ gives rise to three subcases which make the prefunctions~$g_\lambda$, up to the sign of the numerator, denoted by~$\theta := s\ell - rk$.
\bi
\item If~$\theta = 0$, then~$g_\lambda(x) = 1$.

\item If~$s/r > k/\ell$, then~$\theta > 0$\ and~$g_\lambda(x) %
=\delimiterfactor=1200 \Frac{1}{\l(1 - (\ell-k)(k+\ell)\lambda x^{k+\ell}\r)^{\theta/(k+\ell)(\ell-k)}}$.

\item If~$s/r < k/\ell$, then~$\theta < 0$\ and~$g_\lambda(x) %
=\delimiterfactor=1200 \l(1 - (\ell-k)(k+\ell)\lambda x^{k+\ell}\r)^{|\theta|/(k+\ell)(\ell-k)}$.
\ei
Finally, in the symmetric case of~$\ell < k$, the behaviour of~$\theta = s\ell - rk$ is similar to the one examined above, Eept (up to a sign) for the expressions of~$g_\lambda(x)$ within the previous last two subcases. More precisely, when~$\ell < k$ and according to the sign of~$\theta$.
\[
g_\lambda(x) ={}
\begin{dcases} \delimiterfactor=1500 
\b(1 + |\ell-k|(k+\ell)\lambda x^{k+\ell}\b)^{\theta/|\ell-k|(k+\ell)} & \text{if}\ \ \theta > 0,\\ \delimiterfactor=1500 
\b(1 + |\ell-k|(k+\ell)\lambda x^{k+\ell}\b)^{-|\theta|/|\ell-k|(k+\ell)} & \text{if}\ \ \theta < 0.
\end{dcases}\]
Whatever~$k$ and $\ell$ in~$\Z_{>0}$, if~$r = s$ the Lie bracket simplifies 
to~$(\ell-k)x^{\ell+k}\vthx + s(\ell-k)x^{k+\ell}$ and, as~$\theta = s(\ell-k)$, the prefunction reduces 
to~$\b(1 - (\ell-k)(k+\ell)\lambda x^{k+\ell}\b)^{-s/\ell+k}$. This completes the discussion regarding the values of the prefunction factor.

\begin{rem} \label{rem:genliebra}
Turning now to the features of the Lie bracket of the two differential operators~$x^{k}\vthx - rx^{k}$\ 
and~$x^{\ell}\vthx - sx^\ell$ as summarized in Ex.~\ref{ex1}, the discussion runs along the same lines as in the above one for eq.~\eqref{eq:glambda}. Still assuming~$k\neq \ell$ and~$\theta := s\ell - rk$, we have
\be
\item[$(i)$] In the case when the scalar part is $\theta\, x^{k+\ell}$, the expressions of the prefunction~$g_\lambda$ are similar to the ones in the above three cases up to the sign of~$\theta$.
\item[$(ii)$] In the case when the scalar part is $- (rk + s\ell) < 0$, the sign of~$\frac{rk+s\ell}{k-\ell}$ depends only on whether~$k > \ell$ or not. Then, the prefunction turns out to take the general forms
\[ \delimiterfactor=1500
g_\lambda(x) = \b(1 \pm |\ell-k|(k+\ell)\lambda x^{k+\ell}\b)^{\mp |rk+s\ell|/|\ell-k|(k+\ell)}.\]
\ee
\end{rem}
The goal of next~\S\ref{sub:opgra} and following Sections is to adapt the aforegoing parts to the theory of Riordan arrays by regarding the framework of the transformations involved in the integration of the one-parameter group 
in~$\rHW_{\le 1}\!\setminus \! \rHW_0$ (i.e. substitutions with prefunctions) as transformation matrices.

\subsection{Transformation of matrices and Riordan arrays} \label{sub:opgra}
Due to the emphasis placed on generating functions in what follows, we first recall some basic notations (see also Section~\ref{riordangroup} and Appendix~\ref{riordangroup}). 
Usual generating functions~$f(z) = \sum_{n\ge 0} f_n w_n z^n$ are specialized to~$w_n = 1/c_n$, where $(c_n)$ 
($n\in \N$) is a fixed sequence of non-zero constants with~$c_0 = 1$, given once and for all (see Appendix~\ref{app:fpsgf}). In particular, if~$c_n = 1$, $f(z) = \sum_{n} f_n z^n/c_n$ is an \emph{ordinary generating function} (OGF) and, if~$c_n = n!$, $f(z)$ is an \emph{exponential generating function} (EGF). The notation~$\l[z^n\r]$ (or $\langle \frac{z^n}{c_n},f(z)\rangle$) stands for the \emph{coefficient extractor operator}. 
If~$f(z) = \sum_{n} f_n z^n/c_n$, then~$c_n\l[z^n\r]f(z) := f_n$: $\l[z^n\r]f(x)$ denotes for the coefficient of~$z^n$ in the OGF~$f(z)$ and similarly, the coefficient of~$z^n/n!$ in the EGF~$f(z)$ is~$n![z^n]f(z) := f_n$; that is the respective coefficients of~$z^n$ and~$z^n/n!$ in the expansion of~$f(z)$ into powers of~$z$~\cite[p.~19]{FlSe09}). Herein, we are mostly concerned with OGFs and EGFs in~$\C[[z]]$ or~$\R[[z]]$.

\subsubsection{Combinatorics of infinite matrices in $\rHW_\C$ and~$\rHW_{\le 1}\!\setminus \!\rHW_0$} \label{infmat}
Consider, as examples, the upper-left corner of the (doubly infinite) matrices~$M_\omega$, each representing a 
word~$\omega\in \rHW_{\C}$, such as the Pascal matrices or the Stirling matrix (as exemplified in 
Ex.~\ref{ex:genstirling2}, \S\ref{nof}, and Ex.~\ref{ex:binmatstirling}, \S\ref{rgr}). More precisely, 
for~$\omega = a^+a$ ($\omega\map XD$) one gets the array of the usual Stirling numbers of the second 
kind~$\l(\stirp{n}{k}\r)_{n,k\in \N}$ from their classical recurrence relation: once the first line is fixed, the Stirling array can be constructed iteratively, whose bivariate EGF is~$\sum_{k\ge 0} \stirp{n}{k} x^n/n! y^k = %
e^{y(e^x-1)}$.

\begin{lem} \label{lem:rfalg} \emph{(Duchamp \textit{et al.}~\cite{DuPT11,DPSHB04})}\
For any homogeneous operator~$\omega\in \rHW_{\C}$, all lines~$M(n,k)_{k\in \N}$ of the Stirling type matrix are finitely supported. Such matrices are said \emph{row finite} with set denoted as~$\rRFM_\N(\C)$, and their composition makes~$\rRFM_\N(\C)$ into an algebra. 
\end{lem}
\begin{proof}
In each case, the matrix has the form of a staircase where each ``step'' depends on the number~$\delta = |\omega|_a$. One can prove precisely that each row ends with a ``one'' in the cell~$(n,n\delta)$, and we number the entries 
from~$(0,0)$. Thus, all matrices are row finite and unitriangular if, and only if, $\delta = 1$ (the matrices of coefficients for words in~$\rHW_{\le 1}\!\setminus \!\rHW_0$). Moreover, the first column is~$(1, 0, \ldots, 0,\ldots, 0, \ldots)$ if, and only if, $\omega$ ends with an ``a'' (this means that~$\fN(\omega^n)$ has no constant term for all~$n > 0$). 
\end{proof}
Observe that such matrices belong to the group of \emph{unipotent matrices}. As emphasized in \S\ref{intopg} and in Remark~\ref{rem:unipotent} of \S\ref{nof}, unipotent matrices are required for words in~$\rHW_{\le 1}\!\setminus \!\rHW_0$, i.e. which have been proved of special type: the matrices of substitutions with prefunctions. Note also that the matrices in~$\rRFM_\N(\C)$ are coding certain classes of operators, such as the continuous operators in the Fréchet space~$\C[[z]]$ endowed with the (semi-norms) topology of Treves~\cite[p.~43]{DuPT11}.

\subsubsection{The algebra $\cL(\C^\N)$ of sequence transformations} \label{transfseq}
Let~$\C^\N$ be the vector space of all complex sequences, equipped with the Treves product topology. It is easy to check that the algebra~$\cL(\C^\N)$ of all continuous operators~$\C^\N\lra \C^\N$ is the space~$\rRFM_\N(\C)$. 
For a sequence~$A = (a_n)_{n\ge 0}$, the transformed sequence~$B = M A$ is given by~$B = (b_n)_{n\ge 0}$ 
with~$b_n = \sum_{k\ge 0} M(n,k) a_k$. We may associate a series with a given sequence~$(a_n)_{n\in \N}$ and a sequence of prescribed (non-zero) denominators~$(c_n)_{n\in \N}$ to a GF~$\sum_{n\ge 0} a_n z^n/c_n$. Thus, once the~$c_n$'s have been chosen, to every (linear continuous) transformation of generating functions, one can associate a corresponding matrix.

The algebra~$\cL(\C^\N)$ possesses many interesting subalgebras and subgroups, such as the algebra of lower triangular 
transformations~$\cT_\N(\C)$, the group~$\widetilde{\cT}_\N(\C)$ of invertible elements of the latter (which is the set 
of infinite lower triangular matrices with non-zero elements on the diagonal), the subgroup of unipotent transformations~$\cUT_\N(\C)$ (i.e. the set of infinite lower triangular matrices with elements on the diagonal all equal to 1) and its Lie algebra~$\cNT_\N(\C)$, the algebra of locally nilpotent transformations (with zeroes on the diagonal), etc. (see~\cite{DPSHB04}). 

To each matrix~$M(n,k)_{n,k\in \N}\in \rRFM_\N(\C)$, one may associate an operator~$\Phi_M\in \rEnd\b(\C[[x]]\b)$ 
such that the image of~$f = \sum_{k\in \N} a_k x^k/k!\in \C[[x]]$ is defined as
\[
\Phi_M[f](x) = \sum_{n\in \N} b_n x^n/n!\ \quad \text{with}\ \ b_n = \sum_{k\in \N} M(n,k) a_k.\]
Note that if~$\C[[x]]$ is endowed with the structure of Fréchet space of simple convergence of the coefficients (also called Treves topology), each~$\Phi_M$ is continuous. Then, the next proposition states that there exists no other case.

\begin{prop} \label{algisom}
The correspondence~$\varphi : M\lra \Phi_M$\ from~$\rRFM_\N(\C)$ to~$\cL\b(\C[[x]]\b)$ (continuous endomorphisms) is 
one-to-one and linear. Moreover, since~$\Phi_{MN} = \Phi_M\circ \Phi_N$\ and~$\Phi_I = Id_{\C[[x]]}$, $\varphi$ is a vectorial space isomorphism and also an isomophism of algebras.
\end{prop}
The proposition applies immediately to the one-parameter groups~$\l\{e^{\lambda \omega}\r\}$ generated by homogeneous operators~$\omega$ through Bargmann--Fock representation ($a\map D$ and~$a^+\map X$). Indeed, the matrix 
$\varphi^{-1}(\omega)$ is ($E$ denotes the excess of~$\omega$)
\bi
\item Strictly lower triangular when~$E < 0$,
\item Diagonal when~$E = 0$,
\item Strictly upper triangular when~$E > 0$.
\item[] (These matrices are different from the ``generalized Stirling matrices'' defined by eq.\eqref{eq:genstir2}: their non-zero elements are supported by a line parallel to the diagonal.)
\ei
Consequently, $e^{\lambda \omega}$ always admits a representation as a group of operators in an appropriate space. 
For~$E < 0$, one gets each of the polynomial in the spaces~$\C_{\le n}[x]$ (i.e. the sets of polynomials whose degree is less than~$n$). For~$E\ge 0$, one gets~$\C[[x]]$ endowed with the topology of Treves.

\subsubsection{Substitutions with prefunctions} \label{subspref}
Let~$(c_n)_{n\ge 0}$ be a fixed set of denominators. For a generating function~$f$, we consider the transformation
$\Phi_{g,\phi} [f](x) = g(x) f(\phi(x))$, where~$g\in \C[[x]]$ with~$g_0 = g(0) = 1$\ and~$\phi\in x\C[[x]]$. The matrix of this transformation~$M_{g,\phi}$ is given by the transforms of the monomials~$x^k/c_k$, hence
\begin{equation} \label{eq:mon1}
\sum_{n\ge 0} M_{g,\phi}(n,k) \tfrac{x^n}{c_n} = \Phi_{g,\phi} \l[\tfrac{x^n}{c_n}\r] = g(x) \tfrac{\phi(x)^n}{c_n}\,.
\end{equation}
If~$g, \phi \neq 0$ (otherwise the transformation is trivial), we can write
\begin{equation} \label{eq:mon2}
g(x) = a_\ell \tfrac{x^\ell}{c_\ell} + \sum_{r>\ell} a_r \tfrac{x^r}{c_r}\ \quad \tand \quad %
\phi(x) = \alpha_m \tfrac{x^m}{c_m} + \sum_{s>m} \alpha_s \tfrac{x^s}{c_s}
\end{equation}
with~$a_\ell , \alpha_m \neq 0$ and then, by~(\ref{eq:mon1},\ref{eq:mon2}),
\[
\Phi_{g,\phi}\l[\tfrac{x^k}{c_k}\r] = %
a_\ell (\alpha_m)^k \tfrac{x^{\ell+mk}}{c_\ell c_m^k c_k} + \sum_{t>\ell+mk} b_t \tfrac{x^t}{c_t}\,.\]
Therefore, the equivalence~$M_{g,\phi}\in \rRFM_\N(\C)\Longleftrightarrow \phi$ \emph{has no constant term} holds true (and, in this case, $M_{g,\phi}$ is always lower triangular).

The converse is true in the following sense.  Let~$T\in \cL\l(\C^\N\r)$ be a matrix with non-zero first two columns and suppose that the first index~$n$ such that~$T(n,k)\neq 0$ is less for~$k = 0$ than for~$k = 1$ (which is the case, by~\eqref{eq:mon1}, when~$T = M_{g,\phi}$). Set
\[
g(x) := d_0 \sum_{n\ge 0} T(n,0) \tfrac{x^n}{c_n}\ \quad \tand \quad %
\phi(x) := \tfrac{c_1}{g(x)} \sum_{n\ge 0} T(n,1) \tfrac{x^n}{c_n}\,,\]
then~$T = M_{g,\phi}$ if, and only if, $\sum_{n\ge 0} T(n,k) \tfrac{x^n}{c_n} = g(x) \tfrac{\phi(x)^k}{c_k}$ for all~$k$. For EGF ($c_n = n!$) this amounts to restating eq.~\eqref{eq:mon1} as
\[
\sum_{n,k\ge 0} T(n,k) \tfrac{x^n}{n!} y^k = g(x) e^{y\phi(x)},\]
Once integrated, the one-parameter group~$U_\lambda$ reveals the generalized Stirling matrix ($M_{g,\phi} %
= \l(\stirp{n}{k}_{\omega}\r))$ expressed by the following proposition.

\begin{prop} \label{equiv}
For any homogeneous operator~$\omega\in \rHW_{\le 1}\!\setminus \! \rHW_0$, let~$f\map U_\lambda[f]$ be the one-parameter group~$e^{\lambda \omega}$ and let~$M_{g,\phi}\in \rRFM_\N(\C)$ denote the corresponding matrix transformation. Assuming that~$E\ge 0$, where~$E$ denotes the excess of~$\omega$, the following two condition are equivalent for any suitable class of functions~$f$ (e.g.~$f$ analytic on~$(0,1)$).
\begin{flalign*}
i)\qquad \sum_{n,k\ge 0} M_{g,\phi}(n,k) \tfrac{x^n}{n!} y^k &= g(x) e^{y\phi(x)}.\\
ii)\qquad U_\lambda[f](x) = g\l(\lambda x^{E}\r)\,& f\l(x\b(1 + \phi(\lambda x^{E})\b)\r).
\end{flalign*}
\end{prop}
\begin{proof}
From eq.~\eqref{eq:genstir2} in~\S\ref{nof}, one first has the following equality between continuous operators
\[
U_\lambda = e^{\fN(\omega^n)} = \sum_{n,k\ge 0} M_{g,\phi}(n,k)\, %
\tfrac{\lambda^n}{n!}\, x^{n E}\, x^k \partx^k.\]
Assuming condition~\emph{i)}, let us check \emph{ii)} for~$f$ a monomial, namely choose the test functions~$f = x^p$, for~$p=0, 1,\ldots$
\begin{flalign} \label{eq:condii}
U_\lambda(x^p) &= \sum_{n\ge 0} \sum_{k=0}^p M_{g,\phi}(n,k)\, \tfrac{(\lambda x^{E})^n}{n!} \tfrac{p!}{(p-k)!}\,x^p\nonumber\\
&= x^p \sum_{k=0}^p \l([y^k] g(\lambda x^{E}) e^{y\phi(\lambda x^{E})}\r)\tfrac{p!}{(p-k)!}\nonumber\\
&= g(\lambda x^{E}) x^p \sum_{k=0}^p \bin{k}{p}\, \phi(\lambda x^{E})^k %
= g(\lambda x^{E}) \l(x\b(1 + \phi(\lambda x^{E})\b)\r)^p.
\end{flalign}
Now, as both sides in \emph{ii)} are continuous and linear in~$f$ and the set of monomials is total~\cite{Bourbaki03} in the space of formal power series endowed with the  topology of Treves (the usual ultrametric topology would not be enough for~$E = 0$), condition~\emph{ii)} holds for any suitable function~$f$.

Conversely, assume condition~\emph{ii)}, then~$U_\lambda(e^{yx}) = g(\lambda x^{E}) e^{yx\l(1+\phi(\lambda x^{E})\r)}$ and, by eq.~\eqref{eq:condii}, one gets
\[
\sum_{n,k\ge 0} M_{g,\phi}(n,k)\, \tfrac{(\lambda x^{E})^n}{n!} (xy)^k %
= g(\lambda x^{E}) e^{yx\phi(\lambda x^{E})}.\]
Finally, the change of variables ($\lambda x^{E}\ra x$ and~$xy\ra y$) yields condition~\emph{i)}. This proves the required Prop.~\ref{equiv}.
\end{proof}

\begin{rem} \label{rem:sheffer}Eq.~\eqref{eq:mon1} is called the \emph{Sheffer condition}, and it is a fundamental notion in Def.~\ref{defra}, i.e. the definition of Riordan arrays~\cite{Roman05,Stanley99,Wilf06}. From now on, we will suppose that~$\phi$ has no constant term ($\alpha_0 = 0$). Moreover, $M_{g,\phi}\in \widetilde{\cT}$ iff 
$a_0,\, \alpha_1\neq 0$, which implies that (on the diagonal)~$M_{g,\phi}(n,n) = a_0/c_0\; (\alpha_1/c_1)^n$. Hence,
\[
M_{g,\phi}\in \cUT\Longleftrightarrow a_0/c_0 = \alpha_1/c_1 = 1,\]
which, for EGF and OGF, reduces to the equivalence~$M_{g,\phi}\in \cUT\Longleftrightarrow a_0 = \alpha_1 = 1$. In this setting, $g$ and~$\phi$ meet the conditions that~$g\in \C[[x]]$\ with~$g_0 = g(0) = 1$ and~$\phi\in x\C[[x]]$. In classical combinatorics (using OGF and EGF), the matrices~$M_{g,\phi}(n,k)$ are known as \emph{Riordan matrices} 
(see Def.~\ref{defra} and Thm.~\ref{thmrg} in \S\S\ref{gra}--\ref{rgr} and, for example, \cite{Roman05,SGWW91}).
\end{rem}

\section{Riordan arrays and Riordan group} \label{riordangroup}
Since power series (OGF and EGF) play a prominent role in the present paper, some basic useful definitions are recalled below, in the beginning of \S~\ref{sub:opgra} and in Appendix~\ref{app:fpsgf}. (For more details on formal power series, the reader may refer to~\cite[Chaps.~2--7]{BaLa06}, \cite[Chap.~I]{Cartan61}, \cite[App.~A5]{FlSe09}, \cite[Chap.~2]{Roman05}, \cite[Chap.~3]{Sprugnoli06}, etc.)

In the following, $(\cS,\times)$ will denote the group of the formal power series~$g\in \C[[x]]$ with non-zero constant term ($g_0 \ne 0$) equipped with the formal multiplication. $x\cS$ will denote the set of all the series~$f = xg$ such that~$f_0 = 0$ and~$f_1 \ne 0$; the maximal ideal in~$\C[[x]]$ consisting of the series with $f_0 = 0$ also forms a composition semigroup. $(x\cS,\circ)$ is also a group for the formal composition in~$\C[[x]]$ (see the excellent papers~\cite{Bacher06,BaLa06,OFarell08}). Moreover, $\ov{f}$ denotes the inverse of~$f$ for the formal composition in~$(x\cS,\circ)$ and~$\ov{f}$ is named the {\em reverse series} of~$f$ in $\C[[x]]$.  
 
\subsection{Riordan arrays} \label{gra}
\emph{Riordan arrays}, named after John Riordan, were introduced by Shapiro~\emph{et al.} in~\cite{SGWW91} and Roman~\cite{Roman05} to generalize the properties of the Pascal Triangle. Along which, the group structure of the set of Riordan arrays was shown, and called the \emph{Riordan group}. In~\cite{Sprugnoli94}, Sprugnoli made use of Riordan arrays for proving combinatorial identities, and then he showed how Riordan arrays can be used to perform combinatorial sum inversions. These seminal articles were followed by a number of papers, 
e.g.~\cite{Barry09,He11,HeSp09,Hennessy11,MRSV97,Sprugnoli06,WaWa08} exploring classical and generalized Riordan arrays, and the theory of the Riordan group and subgroups.

\begin{defi} \label{defra}
Let~$\b(c_n\b)_{n\in\N}$ be a fixed reference sequence of non-zero constants with~$c_0 = 1$. A generalized \emph{Riordan array} with respect to the sequence~$(c_n)$ is a pair~$\b(g(x),f(x)\b)$ of power series such 
that~$g = \sum_{n} g_n x^n/c_n\in \cS$, i.e.~$g$ is unit (or $g_0\ne 0$), and~$f = \sum_{n} f_n x^n/c_n\in x\C[[x]]$, i.e.~$f_0 = 0$. The Riordan array~$\b(g(x),f(x)\b)$ defines an infinite lower triangular 
array~$T = \b\{d_{n,k}\b\}$ ($n,\, k\in \N$) according to the rule
\begin{equation} \label{eq:arraydnk}
d_{n,k} = c_n\l[x^n\r]g(x)\, \frac{f(x)^k}{c_k},
\end{equation} 
where the functions~$g(x) f(x)^k/c_k$ are referred to as the column GFs (or the GF of the $k$th column) of the Riordan array. Furthermore, if~$f\in x\cS$, i.e.~$f$ is such that~$f_1 = f'(0)\ne 0$ (or~$ord(f) = 1$), the Riordan array is said to be \emph{proper}. 
\end{defi}
By definition, any proper Riordan array is invertible; that is, the infinite lower triangular 
array~$T = \b\{d_{n,k}\b\}$ $(k,\, n\in \N)$ is such that~$d_{n,n}\ne 0$ for all~$n$. Furthermore, for any proper Riordan array~$T = \b(g(x),f(x)\b)$, its diagonal sums are just the row sums of the vertically stretched 
array~$\b(g(x),xf(x)\b)$ and hence have OGF~$\frac{g(x)}{1-xf(x)}$. As regards the sequence~$(c_n)$, ordinary Riordan arrays correspond to the case of~$c_n = 1$, while exponential Riordan arrays corresponds to the case 
of~$c_n = n!$. 

\begin{theorem} \label{thmsumdnk}
Let~$T = \b(g,f\b) = \b\{d_{n,k}\b\}$ ($n,\, k\in \N$) be a proper Riordan array with respect to~$(c_n)$ and 
let~$h(x) = \sum_{n\ge 0} h_n \frac{x^n}{c_n}$ be the generating function of the sequence $(h_n)$. 
Then we have 
\begin{equation} \label{eq:sumdnk}
\sum_{k=0}^n d_{n,k} h_k = c_n\l[x^n\r]g(x)h\b(f(x)\b),
\end{equation}
or equivalently, the so-called ``fundamental theorem of proper Riordan arrays'',
\[
\b(g(x),f(x)\b) * h(x) = g(x) h\b(f(x)\b),\]
where~$*$ denotes the usual array product.
\end{theorem}
The \emph{fundamental theorem of proper Riordan arrays} is a somewhat practical statement of Thm.~\ref{thmsumdnk}. Given an integer sequence~$(\lambda_n)_{n\in \N}$\ with OGF~$\Lambda$, the OGF of the sequence~$\l\{T (\lambda_n)\r\}$ ($n,\, k\in \N$) 
is~$g(x) \Lambda\b(f(x)\b)$, where a sequence is regarded as an (infinite) column vector. The (infinite) array~$T$\ can thus be considered to act on the ring of integer sequences~$\Z^\N$\ by multiplication. This action can be extended to the ring~$\Z[[x]]$ by~$T : \Lambda(x)\map T\Lambda(x) = g(x) \Lambda\b(f(x)\b)$.

\begin{rem}
A useful (combinatorial) alternative definition to the (algebraic) Def.~\ref{defra} of proper Riordan arrays is in terms of the so-called \emph{$A$-sequence} and \emph{$Z$-sequence}. It was found by Rogers in~\cite{Rogers78} and extensively developed after Shapiro \emph{et al.}~\cite{SGWW91} by Merlini \emph{et al.}~\cite{MRSV97}, Sprugnoli~\cite{Sprugnoli94,Sprugnoli06}, etc. Let~$A(x)$ in~$\cS$ ($a_0\ne 0$) and~$Z(x)$ be the two GFs of~$A$-sequences $(a_i)$ and~$Z$-sequences $(z_i)$\ $(i\in \N)$ that satisfy the relations~$f(x) = xA\b(f(x)\b)$ 
and~$g(x) = \frac{g_0}{1 - xZ(f(x))}$. Then, any array~$T = \b\{d_{n,k}\b\}$, for~$0\le k\le n$, is a proper Riordan array if, and only if, every element~$d_{n,k}$ and~$d_{n,0}$ can be expressed as the finite linear combinations for all~$n\ge k\ge 0$,
\begin{equation} \label{AZrec}
d_{n,k} = \sum_{j=0}^{n-k} a_j d_{n-1,k+j-1}\ \quad \text{and}\ \quad  d_{n,0} = \sum_{j=0}^{n-1}  z_j d_{n-1,j}.
\end{equation}
Now, the sequences~$(a_i)$ and~$(z_i)$ are both unique since the~$A$-sequence only depends on~$f$ and the~$Z$-sequence can be constructed by expressing~$z_i$ in terms of the elements in row~$i$ for all~$i\in \N$ in a unique way. Actually, both sums in~\eqref{AZrec} are finite because~$d_{n,k} = 0$ for all $k > n$ by definition. In summary, the $Z$-sequence characterizes column $0$, while the $A$-sequence characterizes all the other columns of the Riordan array. Any proper Riordan array~$T = (g,f)$ as defined in Def.~\ref{defra} is characterized by the triple~$\b(g_0, Z, A\b)$.
\end{rem}

\subsection{The Riordan group} \label{rgr}
The \emph{Riordan group} introduced by Shapiro~\emph{et al.} in~\cite{SGWW91} is algebraically defined by Bacher in~\cite{Bacher06} as the interpolation group, which has a faithful representation into the infinite lower triangular matrices and carries thus the natural structure of a Lie group. From Thm.~\ref{thmsumdnk}, we can further compute the product~$(g,f) * (h,\ell)$ of two Riordan arrays. In fact, the column GFs of~$\b(h(x),\ell(x)\b)$ 
are~$h(x)\,\ell(x)^k/c_k$ by Def.~\ref{defra}. Thus, the~$k$th column GF of the product~$\delimiterfactor=1200 %
\l(g(x),f(x)\r) * \l(h(x),\ell(x)\r)$ is~$g(x) h\b(f(x)\b)\, \ell \b(f(x)\b)^k/c_k$, which means that it is also a Riordan array, i.e.,
\begin{equation} \label{eq:product}
\b(g(x),f(x)\b) * \b(h(x),\ell(x)\b) = \l(g(x) h\b(f(x)\b), \ell \b(f(x)\b)\r).
\end{equation}
For any fixed non-zero sequence~$(c_n)$ with~$c_0 = 1$, the set of all proper Riordan arrays~$\b\{(g,f)\b\}$, is a group under the product of matrices. (Recall that, by Def.~\ref{defra}, a proper Riordan array meets the conditions that~$g\in \cS$ and~$f\in x\cS$, i.e.~$g$ is an invertible series ($g_0\ne 0$), $f_0 = 0$ and $f_1\ne 0$.)

\begin{theorem} \label{thmrg}
For any non-zero reference sequence~$(c_n)$ with $c_0 = 1$, the set of all proper generalized Riordan arrays~$\b\{(g,f)\b\}$ endowed with the Riordan matrix product~$*$ in eq.~\eqref{eq:product} is a group. The identity of this group is~$I = (1,x)$ and the inverse of any array~$(g,f)$ is~$\l(\frac{1}{g\circ \ov{f}}, \ov{f}\r)$, where~$\ov{f}$ is the reverse series of~$f$ (i.e. its inverse for the formal composition in~$\C[[x]]$). This group is called the \emph{Riordan group} with respect to~$(c_n)$, denoted by~$\b(\cR,*\b)$ or~$\cR$. 
\end{theorem}
\begin{proof} According to eq.~\eqref{eq:product}, the Riordan group~$\cR$ is closed under the Riordan product and the product is associative. The matrix~$(1,x)$ is an element of~$\cR$ and, for each matrix~$\b(g(x),f(x)\b)\in \cR$, there exists a matrix $\delimiterfactor=1200 \b(\dfrac{1}{g(\ov{f}(x))},\ov{f}(x)\b)$ in~$\cR$ for which we have
\begin{align*}
\b(g(x),f(x)\b) &* (1,x) = \b(g(x),f(x)\b) = (1,x) * \b(g(x),f(x)\b)\ \qquad \text{and}\\
\b(g(x),f(x)\b) &*\delimiterfactor=1200 \b(\frac{1}{(g\circ \ov{f})(x)},\ov{f}(x)\b) = (1,x) %
=\delimiterfactor=900 \ov{f}(x) * \b(g(x),f(x)\b).
\end{align*}
Hence, the inverse of a matrix~$L = \b(g(x),f(x)\b)$ under the matrix product in~$\cR$ is
\begin{equation} \label{eq:inv}
L^{-1} = \b(g(x),f(x)\b)^{-1} =\delimiterfactor=1200 \b(\frac{1}{(g\circ \ov{f})(x)},\ov{f}(x)\b),
\end{equation}
which proves the assertion of the theorem.
\end{proof}
Note that, by eq.~\eqref{eq:arraydnk}, the coefficients~$d_{n,k}$ of the matrix~$(1,x)$ write 
down~$d_{n,k} = c_n \l[x^n\r] x^k/c_k = c_n/c_k \l[x^{n-k}\r] 1 = \delta_{n,k}$, where~$\delta_{n,k}$ is the Kronecker delta (defined by~$\delta_{n,n} = 1$ and~$\delta_{n,k} = 0$ for~$n\ne k$).

\begin{rem} \label{rem:itmat}
A large number of infinite lower triangular arrays are Riordan arrays. Particularly, the \emph{iteration matrices} are in this case. To every power series~$f = \sum_{n\ge 0} f_n x^n/c_n$, we can associate the infinite lower iteration matrix (with respect to~$(c_n)$)~$B(f) := \b\{B_{n,k}\b\}$, 
where~$B_{n,k}^{(c_n)} = B_{n,k}\b(f_1, f_2,\ldots\b)$ is the Bell polynomial with respect to~$c_n$, which is defined as~$f(x)^k/c_k := \sum_{n\ge k} B_{n,k} x^n/c_n$ (see Comtet~\cite[p.~145]{Comtet74}). Therefore, 
$B_{n,k} = c_n\l[x^n\r]f(x)^k/c_k$, which implies that the iteration matrix~$B(f)$ is the Riordan array~$(1,f)$. Now, the next important property of the iteration matrix~\cite[p.~145, Thm.~A]{Comtet74},
\[ 
B\l(f\b(g(x)\b)\r) = B\b(g(x)\b) * B\b(f(x)\b),\]
is trivial in the context of the theory of Riordan arrays for~$(1,f\circ g) = (1,g) * (1,f)$. The well-known Fa\`a di Bruno formula is a specialization of the summation rule in~eq.~\eqref{eq:arraydnk} (Def~\ref{defra}):
\[
\sum_{k=0}^n B_{n,k} \b(g_1, g_2, \ldots, g_{n-k+1}\b) f_k = n!\l[x^n\r]f\b(g(x)\b).\]
In addition, for any series~$f\in \C[[x]]$ such that~$f_1 = 1$,~$\l(1,\ov{f}\r) = B\l(\ov{f}\r)$ is also an iteration matrix. Hence, the set of iteration matrices with respect to~$(c_n)$ is a non-empty subset of the Riordan group~$\cR$ (with respect to~$(c_n)$); it is also closed under multiplication and inversion in~$\cR$ and, as such, it is a subgroup of~$\cR$.
\end{rem}
\begin{rem} In Shapiro \emph{et al.}~\cite{SGWW91}, the Riordan group is a unipotent group of infinite lower-triangular integer matrices that are defined by a pair of power series~$(g,f)$ such that~$g\in \cS$ (i.e, $g\in 1 + x\C[[x]]$) and~$f\in x\cS$ with~$f_1 = 1$. Such is a stronger definition than usually required for the Riordan group. However, in the present paper, we make only use of {\em proper} Riordan arrays whose diagonal elements equal~$1$. The series~$f = \sum_{n\ge 1} f_n x^n/c_n$ in~$x\cS$ will henceforth meet also the additional property that~$f_1 = 1$, which naturally suppose infinite lower-triangular integer matrices with diagonals~$d_{n,n} = 1$ (i.e. the unipotent Riordan group).
\end{rem}

\begin{ex} \label{ex:binmatstirling}
Let the matrix~$P$ of the Riordan group be defined by the pair~$\b(\frac{1}{1-x},\frac{x}{1-x}\b)$ of OGFs. The coefficients of~$P$ are~$d_{n,k} = \l[x^n\r]\frac{1}{1-x}\, \b(\frac{x}{1-x}\b)^k %
= \l[x^n\r]\frac{x^k}{(1-x)^{k+1}} = \bin{n}{k}$, and~$P$ is the well-known Pascal (or binomial) matrix. The inverse of~$P$ in~$\cR$ is $P^{-1} = \b(\frac{1}{1+x},\frac{x}{1+x}\b)$, with coefficients~$(-1)^{n-k} \bin{n}{k}$. In the Riordan group, for~$m$ integer, $P^m = \b(\frac{1}{1-mx},\frac{x}{1-mx}\b)$, the general term of which 
is~$m^{n-k} \bin{n}{k}$, and the inverse~$P^{-m}$ is given by~$\b(\frac{1}{1+mx},\frac{x}{1+mx}\b)$. The row OGFs of~$P$ are~$(x+1)^n$, and the row OGFs of~$P^{-1}$ are~$(x-1)^n$. In addition, $\b(P * (1,-x)\b)^2 = I$, where~$(1,-x)$ is the diagonal matrix with alternating~$1$'s and~$-1$'s on the diagonal; so, $P * (1,-x)$ is said to have ``order two'' and~$P$ to have ``pseudo-order two'', which also means that~$P$ is a pseudo-involution.\par
\no Similarly, let the Pascal matrix~$\tilde{P}$ in~$\cR$ defined by the exponential Riordan array~$\l(e^x,x\r)$. The coefficients of~$\tilde{P}$ are~$n!\l[x^n\r] e^x\, \frac{x^k}{k!} %
= \l[x^n\r]\sum_{n\ge k} \frac{n!}{k!} \frac{x^n}{(n-k)!} = \bin{n}{k}$. Moreover, $\tilde{P}^m$ is the element~$\l(e^{mx},x\r)$ of~$\cR$, and the inverse~$\tilde{P}^{-m}$ of~$\tilde{P}^m$ is~$\l(e^{-mx},x\r)$. The row EGFs of~$\tilde{P}$ are~$e^x x^n/n!$ and the row EGFs of~$\tilde{P}^{-1}$ are~$(-1)^n e^{-x} x^n/n!$. In addition, since each Riordan array~$\b(g(x),f(x)\b)$ has the bivariate OGF~$\frac{g(x)}{1 - yf(x)}$, $P$ has the bivariate 
OGF~$\frac{1}{1 - x(1+y)}$ and~$\tilde{P}$ has the bivariate EGF~$g(x)e^{yf(x)}$.

As another example, let us consider the exponential Riordan array~$\b(1,\ln(1+z)\b)$. The general term of the columns EGF is~$n!\l[z^n\r] \b(\ln(1+z)\b)^k/k! = \stirc{n}{k}$ (or also~$s(n,k)$): the usual Stirling number of the first kind. The row EGFs are~$\sum_{k=0}^n \stirc{n}{m} z^k = z^{\und{n}} = z(z-1)(z-2)\cdots (z-n+1)$ (the ``falling factorial'' polynomials) and form the sequence associated to~$e^z - 1$~\cite[p.~206 \& p.~212]{Comtet74}.\par
The inverse of the array~$\b(1,\ln(1+z)\b)$ is~$\l(1,e^z-1\r)$, whose general term of the columns EGF 
is~$n!\l[z^n\r] \l(e^z-1\r)^k/k! = \stirp{n}{k}$ (or also $S(n,k)$): the usual Stirling number of the second kind. The polynomials $\varpi_n(x) = \sum_{k=0}^n \stirp{n}{m} x^k$ are referred to as the exponential polynomials. The sequence~$\l\{\varpi_n(x)\r\}$ is associated to~$\ln(1+z)$.
\end{ex}

\section{Riordan subgroups} \label{riordansubgroups}
A number of basic subgroups of the Riordan group~$\cR$ can be found e.g. in~\cite{DaSW12}. Among others, the set of Riordan arrays~$\b\{(g(x),x)\b\}$, where~$g\in \cS$, is also a subgroup of~$\b(\cR,*\b)$. It is called the Appell subgroup of~$\cR$ and denoted by~$\cA$. Still assuming~$g\in \cS$ and~$f\in x\cS$, the \emph{Bell} subgroup of~$\cR$ is defined by~$\cB = \b\{(g,xg)\b\}$ and the \emph{Lagrange} (or \emph{associated}) subgroup of~$\cR$ is given by~$\cL = \b\{(1,f)\b\}$. Furthermore, for all~$(g,f)\in \cR$, Thm.~\ref{thmrg} implies
\begin{flalign*} 
\b(g(x),x\b) * \b(f(x),xf(x)\b) & =\delimiterfactor=1200 \b(\frac{xg(x)}{f(x)},x\b) * \b(\frac{f(x)}{x},f(x)\b) %
=\delimiterfactor=900  \b(g(x),f(x)\b)\\ 
\text{and}\ \quad & \b(g(x),x\b) * \b(1,f(x)\b) = \b(g(x),f(x)\b).
\end{flalign*}
Since~$\cA$ is normal, $\cR\simeq \cA \rtimes \cB$ and $\cR = \cA \rtimes \cL$: the Riordan group is a semi-direct product of the Appell normal subgroup and the Lagrange and Bell subgroups, respectively. 

In what follows, these two classic subgroups of~$\cR$, along with another subgroup which generalizes~$\cL$ 
and~$\cB$, that is the set of proper Riordan arrays of the form~$\l\{\b(g(x)^\rho,xg(x)\b)\r\}$ with~$\rho$ rational (and, by definition, $g\in \cS$) are under consideration. The latter group (first observed in~\cite{ChKS10} 
when~$\rho$ is a positive integer) is defined in next Def.~\ref{powerrhosubgp} and plays a key role in Sections.~\ref{groupsrlm} and~\ref{extlaw}. (See e.g., Def.~\ref{def2} in~\S\ref{stripedmat}.)

\begin{defi} \label{powerrhosubgp}
By Thm.~\ref{thmrg}, the set~$\l\{\b(g^\rho,xg\b)\r\}$ with $\rho\in \Q$ of proper Riordan arrays is a subgroup of~$\cR$ with respect to the Riordan product; it is denoted by~$\b(\cG(\rho),*\b)$. 
\end{defi}
More precisely, the identity in~$\cG(\rho)$ is~$I = (1,x)$. Now, let~$h(x)\in x\cS$ be the reverse series 
of~$xg(x)$, that is~$h(x) = \frac{x}{(g\circ h)(x)}$. Since~$g\in \cS$, any matrix~$L = \l(g^\rho,xg\r)$ has a unique inverse~$L^{-1} =\delimiterfactor=800 \l(\frac{1}{g^\rho\circ h},h\r)$ in~$\cG(\rho)$. When~$\rho$ is specialized to~$0$ and~$1$, respectively, the Lagrange and the Bell subgroups of~$\cR$ are obtained as subgroups 
of~$\b(\cG(\rho),*\b)$. 

\subsection{Riordan subgroups and~$\rHW_{\le 1}\!\setminus \! \rHW_0$} \label{subgroups}
From Section~\ref{onepargroup}, the real parameter~$\lambda$ is assumed by definition to be close to 0: 
$|\lambda| < 1$.\par 
First, let~$g_{\lambda}(x) = \frac{1}{(1-\lambda x)}$ and~$g_{m\lambda}(x) = \frac{1}{(1-m\lambda x)}$. 
Since~$\lambda\ne \pm 1$, eq.~\eqref{eq:powerm} below holds for any such~$\lambda$ and integer~$m$, by definition of the Riordan product~$*$ in Thm.~\ref{thmrg},
\begin{equation} \label{eq:powerm}
\delimiterfactor=1200 L_{\lambda}^{m} := \l(g_{\lambda}(x),xg_{\lambda}(x)\r)^m %
= \l(g_{m\lambda}(x),xg_{m\lambda}(x)\r) =: L_{m\lambda}.
\end{equation}
In eq.~\eqref{eq:powerm}, the identity~$I = (1,x)$ is obtained for~$m = 0$; if~$m = -1$, 
then~$L_{\lambda}^{-1} = \b(\ov{g}_\lambda,x\ov{g}_\lambda\b)$, where~$\ov{g}_\lambda$ is the reverse series 
of~$g_\lambda$ in~$\cS$, written $\ov{g}_\lambda(x) = \frac{1}{(1+\lambda x)}$ by Thm.~\ref{thmrg}. 
The Riordan product of matrices~$L_{\lambda}$, also meets the property~$L_{\lambda_1} * L_{\lambda_2} %
= L_{\lambda_1 + \lambda_2}$.

Now, consider the Riordan group associated to the operator~$q(x)\partx = x\vthx$. For~$|\lambda| < 1$, the two 
matrices~$L_{\lambda}$ and its inverse~$L_{\lambda}^{-1}$ are the generators of the Riordan subgroup~$\Hom$ of homographies. Since~$L_{\lambda}^m = L_{m\lambda}$ and~$L_{-\lambda}^m = L_{-m\lambda}$, $\b(\Hom,*\b)$ is also a group with two generators.\par
Next, consider the Riordan group associated to the operator~$q(x)\partx + v(x) = x\vthx + x$. 
By continuity of~$\frac{1}{(1\pm m\lambda x)}$ at~$\lambda = \pm 1$, $\Lim_{\lambda\to 1-}$ $\frac{1}{(1 - m\lambda x)}$ and~$\Lim_{\lambda\to -1+}$ $\frac{1}{(1 + m\lambda x)}$ have both a meaning for~$m$ integer. Hence, $(\Hom,*)$ is also defined in the limits for~$|\lambda| \le 1$.

For example, the Riordan Pascal matrix~$P = \l(\frac{1}{1-x},\frac{x}{1-x}\r)$ is well defined as a substitution with prefunction, and~$P,\, P^{-1}$ are the two generators of the Bell Riordan subgroup~$\cB$. So, for any integer~$m$, $P^m = \l(\frac{1}{1-mx},\frac{x}{1-mx}\r)$ (see Ex.~\ref{ex:binmatstirling}).

\subsection{Automorphic Riordan subgroups and Puiseux series} \label{autoriordan}
As defined in Appendix~\ref{app:fpsgf}, if~$\K$ is a field then the field of Puiseux series with coefficients in~$\K$ is defined informally as the field~$\puiseux{\K}{x}$ of formal Laurent series~$\laurent{\K}{x}$ of the form~$\sum_{n\ge k} a_n x^{n/N}$, where~$N\in \Z_{>0}$ and $k\in \Z$. In other words, $\puiseux{\K}{x}$ with coefficients in~$\K$ is~$\dis \bigcup_{N\in \Z_{>0}}\!\! \delimiterfactor=800 \K\l(\l(x^{1/N}\r)\r)$, where each element of the union is a field of formal Laurent series over~$x^{1/N}$ (considered as an indeterminate). In the following, we set~$\K = \C$ and~$\puiseux{\C}{x}$ denotes the field (and also the vector space over~$\C$) of Puiseux series with complex coefficients.

\begin{defi} \label{riordantranf}
Let~$\cM = \l\{\mu_\rho,\, \rho \in \Q\r\}$ be the subgroup of the 
group~$\delimiterfactor=1200 \rGL\l(\puiseux{\C}{x}\r)$ of linear transformations of~$\puiseux{\C}{x}$ such that for any~$U\in \puiseux{\C}{x}$, $\mu_\rho \b(U(x)\b) = x^\rho U(x)$ and, similarly, 
$\mu_\rho^{-1} \b(U(x)\b) = x^{-\rho} U(x)$. Further, $\cM$ is the group of transformations of the Riordan group~$\cR$, and so it is for the Riordan subgroups of~$\cR$ such as~$\cG(\rho)$. 
\end{defi}
The usual Riordan product~$*$ is thus extended from~$\cR \times \C[[x]]$ to the product~$\cdot$ 
for~$\cR \times \puiseux{\C}{x}$. More precisely, $\b(f(x),g(x)\b) \cdot U(x) = f(x) U\b(g(x)\b)$ for any matrix~$(f,g)$ in $\cR$ and any power series~$U\in \puiseux{\C}{x}$. Now, since~$U(x) = \Sum_{n\ge k} a_n x^{n/N}$, we have~$(f,g) \cdot U = f(x) \Sum_{n\ge k} a_n g(x)^{n/N}$, where~$k\in \Z$, $N\in \Z_{>0}$ and $a_n\in \C$.\par
By combining the actions of groups in Def.~\ref{riordantranf}, we can construct the diagram in Fig.~\ref{autom}. 
\begin{figure}[!ht]
\[
\xymatrix@!R=.2in{
U(x) \ar@{|->}[r]^-{\mu_\rho} \ar@{|->}[d]_-{\varphi} & x^\rho U(x) \ar@{|->}[d]^{\psi} \\
**[l]g(x)^\rho U\b(xg(x)\b) & 
**[r]x^\rho g(x)^\rho U\b(xg(x)\b) \ar@{|->}[l]^-{\mu_{\rho}^{-1} } }\]
\vskip -.2cm \caption{Automorphic Riordan subgroups \emph{via} their actions on~$\puiseux{\C}{x}$}
\label{fig:autom}
\end{figure}
In this diagram, for any~$L = \b(g(x)^\rho,xg(x)\b)\in \cG(\rho)$ ($\rho\in \Q$), $\varphi$ and~$\psi$ stand respectively for two left actions of~$\cG(\rho)$ on~$\puiseux{\C}{x}$: $\varphi$ is the 
action~$L\cdot U(x) = g(x)^\rho U\b(xg(x)\b)$ and~$\psi$ is the action~$L\cdot V(x) = x^\rho g(x)^\rho U\b(xg(x)\b)$ (where the product~$\cdot$ is the extension of the product~$*$ for~$\cR \times \puiseux{\C}{x}$). From the two actions of the groups~$\cG(\rho)$ on~$\puiseux{\C}{x}$ for all~$\rho\in \Q$, the subgroups of~$\cG(\rho)$ are determined by the outer automorphisms~$\rOut\b(\cG(\rho)\b) = \rAut\b(\cG(\rho)\b)\b/\rInn\b(\cG(\rho)\b)$ which have the 
form~$\mu_\rho^{-1}\,\circ \psi\,\circ \,\mu_{\rho} = \varphi$.

\begin{lem} \label{autom}
For any~$\rho_1,\, \rho_2 \in \Q$, there exists a unique $\mu_{\rho_1}$ such that the action on the 
set~$\puiseux{\C}{x}$ of~$\mu_{\rho_1}^{-1}\circ \b(g(x)^{\rho_2} , xg(x)\b)\circ \mu_{\rho_1}$ is equivalent to the action~$\b(g(x)^{\rho_1},xg(x)\b)\cdot U(x) = \b(g(x)^{\rho_1+\rho_2},xg(x)\b)$.
\end{lem}
\begin{proof}
Let~$U(x)\in \puiseux{\C}{x}$. For any~$\mu_{\rho_1} \in \cM$, we can associate to~$U(x)$ the Puiseux 
series~$V(x) = x^{\rho_1}U(x)$. The left action~$\psi$, $\b(g(x)^{\rho_2},xg(x)\b)\cdot V(x)$, results in the Puiseux series~$W(x) = g(x)^{\rho_2} \b(xg(x)\b)^{\rho_1} U\b(xg(x)\b) %
= x^{\rho_1} g(x)^{\rho_1+\rho_2} U\b(xg(x)\b)$.\par
Now, $\mu_{\rho_1}^{-1}\b(W(x)\b) = g(x)^{\rho_1+\rho_2} U\b(xg(x)\b)$ is equal to the left action~$\varphi$: 
$\b(g(x)^{\rho_1},xg(x)\b)\cdot U(x) = g(x)^{\rho_1} U\b(xg(x)\b)$. This completes the proof.
\end{proof} 
To summarize, for all~$\rho_1,\, \rho_2\in \Q$ the subgroups~$\b(g(x)^{\rho_1+\rho_2},xg(x)\b)$ 
and~$\b(g(x)^{\rho_1},xg(x)\b)$ of~$\cG(\rho)$ are pairwise (outer) automorphic Riordan subgroups.

\begin{cor}
For any~$\rho \in \Q$ there exists a unique application~$\mu_\rho \in \cM$ such that the action on the 
set~$\puiseux{\C}{x}$ of~$\mu_{\rho}^{-1}\circ \b(g(x)^{\rho},xg(x)\b)\circ \mu_\rho$ is equivalent to the action 
of~$\b(g(x)^\rho,xg(x)\b)$  
\end{cor}
\begin{proof}
Immediate from Lemma~\ref{autom} (by specifying~$\rho_1 = 0$ and~$\rho_2 = \rho$).
\end{proof}

\begin{rem}
Let~$U\in \cS$ and~$(g,f)$ be a Riordan matrix in~$\cR$. By Thm.~\ref{thmsumdnk}, $(g,f) * U = gU(f)$, and~$\cR$ is acting on~$\cS$ so that~$g(x)U\b(f(x)\b) = g(x)\sum_{n\ge 0} \l[x^n\r]U(x)\, f(x)^n$. 
Let~$\l\{\b(1,xg(x)\b)\r\}$ be a subgroup of~$\cL$; it is a pure substitution subgroup that transforms the series~$U(x)$ into the series~$U\b(xg(x)\b)$.\par 
In the same way as power series in~$\cS$ provide matrix representations of linear transformations of~$\C[[x]]$, 
there also exist representations of linear transformations of strict Puiseux series, 
in~$\rGL\b(\puiseux{\C}{x}\!\setminus \!\C[[x]]\b)$ .
\end{rem}

\section{Striped Riordan subgroups} \label{groupsrlm}
Again at this point, we make use of the subgroup~$\cG(\rho)$ of~$\cR$ defined in Def.~\ref{powerrhosubgp}, Section~\ref{riordansubgroups}. In the context of striped Riordan subgroups of~$\cG(\rho)$, we will extensively refer to power series (or GFs) restricted to the form~$g_n(x) = \frac{1}{\sqrt[\uproot{2} n]{1 - n\lambda x^n}}$ with integer~$n > 0$. The families of function~$g_n\in \cS$\ and $xg_n(x)\in x\cS$ are actually replacing the usual prefunctions~$g_\lambda$ with substitutions factor~$s_\lambda$ arising in the context of Bargmann--Fock representation in~\S\ref{bfr} (Section~\ref{diffops}), and of one-parameter groups in Section~\ref{onepargroup}.

\subsection{Striped Riordan matrices and subgroups of $\cR$} \label{stripedmat} 
The next two definitions (both given in~\cite{ChKi13}) are at the basis of the notions required in the 
present Section~\ref{extlaw}.

\begin{defi} \label{def1}
Any Bell matrix~$\b(g(x),xg(x)\b)$ with $g\in \cS$ and $g_0 = 1$, such that there exists an integer~$\nu\ge 1$ 
with~$h(x) = \sum_{n} g_n x^{\nu n} = 1 + g_1 x^\nu + g_2 x^{2\nu} + g_3 x^{3\nu} +\cdots$ will be referred to as 
a \emph{$\nu$-striped Riordan matrix}. The number~$\nu$ represents the size of blocks made of~$\nu - 1$ consecutive zeros in each column of the matrix (whence the names ``$\nu$-striped'' or ``$(\nu-1)$-aerated''): in other words, the matrix coefficients meet the condition~$d_{n,k} = 0$ when~$n - k\not \equiv 0\pmod \nu$ for all~$k\le n$.
\end{defi}
The set of $\nu$-striped Riordan arrays is similar to~$\cB$, with stronger properties however: on each array diagonal, 
$d_{n,n} = 1$ and stripes of~$\nu - 1$ consecutive zeros may appear along the columns. In addition, this set is a unipotent subgroup of~$\cB$ under the Riordan product, since~$g_0 = 1$, $(1,x)$ is the identity and~$(g,xg)^{-1} = %
\b(\frac{1}{g\circ (\ov{xg})},\ov{xg}\b)$, is the unique inverse of~$(g,xg)$.

\begin{defi} \label{def2}
Any~$n$-striped Riordan array~$\b(g_n^\rho,xg_n\b)$ with~$\rho\in \Q$ such that the EGF~$g_n$ is of the 
form~$g_n(x) = \frac{1}{\sqrt[\uproot{2} n]{1 - n\lambda x^n}}$ for~$n$ positive integer will be referred to as 
a \emph{$\rho$-prefunction $n$-striped Riordan array}.\par
We let~$\cG(n,\rho)$ stand for the set of the~$\rho$-prefunction $n$-striped Riordan arrays. Further, according 
to Def.~\ref{powerrhosubgp}, $\b(\cG(n,\rho),*\b)$ denotes the particular subgroup~$\l\{\b(g_n^{\rho}(x),xg_n(x)\b)\r\}$ of~$\b(\cG(\rho),*\b)$ for any integer~$n > 0$ and rational~$\rho$.
\end{defi}

\begin{nota} \label{nota} In this section, $g_n\in \cS$ and $xg_n := s_n\in x\cS$, with~$m,\, n\in \Z$ and $n\ne 0$. The~$m$-th compositional power of~$g_n$ is denoted~$g_n^{(\circ m)} := g_n^{(m)}$, defined as
\[
g_n^{(m)} := g_n\circ \cdots \circ g_n\ \ \text{if}\ \ m\in \Z_{\ge 0}\ \quad \text{and}\ \quad %
g_n^{(-m)} := g_n^{(-1)}\circ \cdots \circ g_n^{(-1)}\ \ \text{if}\ \ m\in \Z_{\le 0}.\]
For simplicity of notation, we let~$g_{mn} := g_{n}^{(m)}$ and $s_{mn} := s_{n}^{(m)} = xg_{mn}$. 
Subsequently, by Def.~\ref{def1},
\[
g_{mn}(x) := \b(1 - mn\lambda x^n\b)^{-1/n}\quad \text{and} \quad s_{mn}(x) := x\b(1- mn\lambda x^n\b)^{-1/n}.\]
Similarly, all the current matrices in~$\cG(n,\rho)$ defined in Def.~\ref{def2}, $\l(g_n,s_n\r)$, 
$\b(g_{n}^{(m)},s_{n}^{(m)}\b) := \l(g_{mn},s_{mn}\r)$, $\delimiterfactor=700 \l(\b(g_{mn}^\rho\b),s_{mn}\r)$ 
($\rho$ integer or rational power), etc. will be respectively denoted by~$L_{n}$, $L_{n}^{(m)} := L_{mn}$, 
$L_{mn}^\rho := L_{mn,\rho}$, etc., according to the statements (e.g. $L_{mn,0} := \b(L_{n}^{(m)}\b)^{0}$ will stand 
for~$\b(1,s_{mn}\b)$).
\end{nota}
From Def.~\ref{powerrhosubgp}, the identity in~$\cG(n,\rho)$ is $I = (1,x)$ and the inverse of any 
matrix~$\b(g_n^\rho,s_n\b)$ is~$\delimiterfactor=800 \l(\frac{1}{g_n^\rho\circ \ov{s_n}},\ov{s_n}\r)$, where~$\ov{s_n}\in x\cS$ is the reverse series of~$s_n$. Furthermore, the group~$\l(\cG(n,\rho),*\r)$ ($n\in \Z_{>0}$ and~$\rho\in \Q$) is generated by the Riordan matrices
\begin{align} \label{lnrho} 
L_{n,\rho} & := \b(g_n^\rho(x), s_n(x)\b) ={} %
\l(\frac{1}{(1 - n\lambda x^n)^{\rho/n}} , \frac{x}{(1 - n\lambda x^n)^{1/n}}\r)\ \qquad \text{and}\nonumber\\ 
L_{n,\rho}^{-1} & := \l(\frac{1}{g_n^\rho(x)\circ \ov{s_n}(x)} , \ov{s_n}(x)\r) ={}  %
\l(\frac{1}{(1 + n\lambda x^n)^{\rho/n}} , \frac{x}{(1 + n\lambda x^n)^{1/n}}\r).
\end{align}
Extending eq.~\eqref{lnrho} to~$L_{mn,\rho}$ and~$L_{mn,\rho}^{-1}$, we can fix the values of $n\in \Z_{>0}$ 
and $\rho\in \Q$; then for all $m\in \Z$, the set of functions~$\l\{g_{mn}^\rho\r\}$, lay at the basis of the infinite group~$\cG(mn,\rho)$ for each one given~$m$.

\begin{rem} \label{remgn1}
The designation of~$\cG(n,1)$ as the ``prefunction $n$-striped Riordan subgroup'' of~$\cB$ simply follows the definition of an $n$-striped Riordan matrix, restricted to the usual prefunction~$g_n$ designed in Section~\ref{onepargroup}. Indeed, let~$\gamma\in \Q$ and~$p\in \N$, and consider for example the real linear combination of operators in~$\rHW_{\le 1}\!\setminus \!\rHW_0$ of the form
\[
\gamma\b(a^{+^{(k+1)}}a + pa^{+^k}\b) - (\gamma-1)\b(a^{+^{(k+1)}}a\b) %
= a^{+^{(k+1)}}a + p\gamma a^{+^k}.\]
Then~$p\gamma\in \Q$ and plays the role of the rational~$\rho$ in the subgroup~$\cG(\rho)$ (see Def.~\ref{powerrhosubgp})\par
In the cases of~$\rho = 0$ and~$\rho = 1$, the subgroup~$\cG(n,\rho)$ of~$\cG(\rho)$ reduces respectively to the subgroup~$\cG(n,0)$ of~$\cL$ and the subgroup~$\cG(n,1)$ of~$\cB$ defined in Def.~\ref{def1}. Moreover, 
the group~$\cG(n,1)$ generates the group which contains every~$\cG(n,m)$ with~$m\in \Z$ as subgroup. Notice also that, in the formal topology defined in Appendix~\ref{app:fpsgf}, $\Lim_{n\to 0+} g_n(x) = e^{\lambda}$ 
and~$\Lim_{n\to 0+} \ov{g}_n(x) = e^{-\lambda}$ for all~$|x| < 1/(n\lambda)^{1/n}$. Thus, both 
matrices~$\b(e^{\pm\lambda},e^{\pm\lambda}x\b) = e^{\pm\lambda} I$ are distinct from~$I$ (unless~$\lambda = 0$), but also belong to~$\cR$. Although $e^{\pm\lambda} I$ are not in~$\cG(n,\rho)$, they are the two generators of the subgroup~$\l\{e^{\pm\lambda} I\r\}$ of~$\cR$ for any real~$|\lambda| < 1$.
\end{rem} 

\subsection{The Lie bracket in Riordan subgroups of $\cG(n,\rho)$} \label{lbgroup} 
Under the same assumptions as in the above eq.~\eqref{lnrho}, the following generalization of the statements 
in~\S\ref{Liebprefun} to the group~$\cG(n,\rho)$ is now under consideration. Let~$k, \ell$ and~$r$, $s$ 
in~$\Z_{>0}$ represent the weights of the two differential operators polynomials~$x^{k}\vthx  + rx^{k}$\ 
and~$x^{\ell}\vthx + sx^\ell$. By Eq~\eqref{eq:glambda} in~\S\ref{Liebprefun}, the Lie Bracket of the above operators is
\begin{equation} \label{eq:lb}
\l[x^{k}\vthx + rx^{k} , x^{\ell}\vthx + sx^\ell\r] = (\ell-k) x^{\ell+k} \vthx + \b(s\ell - rk\b) x^n.
\end{equation}
Notice that the evaluation of the Lie bracket of the differential operators~$x^{k}\vthx \pm rx^{k}$ is discussed in Remark~\ref{rem:genliebra} and, as regards~$\cG(n,\rho)$, in Remark~\ref{rem:genliebragnrho} below. 

From notations~\ref{nota}, let~$m := \ell - k\in \Z$, $n := k + \ell > 0$, $g_{mn}(x)=\b(1-mn\lambda x^n\b)^{-1/n}\in \cS$ and~$s_{mn}(x) = x\b(1 - m n \lambda x^n\b)^{-1/n}\in x\cS$. By the Lie bracket in eq.~\eqref{eq:lb}, the prefunction~$\delimiterfactor=700 g_{mn}^\rho := \l(g_{n}^{(m)}\r)^\rho$ and the substitution function~$s_{mn}$\ are then
\begin{equation} \label{eq:gstheta} \delimiterfactor=1200
g_{m n}^\rho(x) ={} \frac{1}{\l(1 - m n \lambda x^n\r)^{\rho/n}}\ \qquad \text{and}\ \qquad %
s_{mn}(x) ={} \frac{x}{\l(1 - m n \lambda x^n\r)^{1/n}},
\end{equation}
where~$\rho = \theta/m\in \Q$ is expressed here in terms of the integer~$\theta := s\ell - rk$.

\begin{rem} \label{rempairLkl}
For all~$k,\,\ell\in \Z_{>0}$ and~$\rho\in \Q$, the subgroup~$\cG(mn,\rho)$ of~$\cR$ under the Riordan product~$*$ is generated by the two prefunction~$mn$-striped Riordan matrices~$L_{mn,\rho} ={} \b(g_{m n}^{\rho} , s_{m n}\b)$ 
and\ $\delimiterfactor=1200 L_{mn,\rho}^{-1} = \b(\frac{1}{g_{m n}^{\rho}\circ \ov{s_{m n}}} , \ov{s_{m n}}\b)$, 
where~$mn = (\ell - k)(k + \ell)\in \Z$ and~$\rho = \theta/m\in \Q$. The case when~$\rho = 1$ leads to the three subcases investigated in~\S\ref{subst2}. Similarly, when~$\ell > k$ and~$r = s$, $\theta = rm$ and the 
subgroup~$\cG(mn,\rho)$ reduces to a particular case which is treated in next Ex.~\ref{ex:r=s}.
\end{rem}
let~$\theta := s\ell - rk\in \Z$ and $mn = (\ell - k)(k + \ell)$ (under the current assumptions 
on~$k,\, \ell,\, r,\, s$). The present discussion about the expression of~$g_{m n}^\rho$ follows the same lines as in~\S\ref{Liebprefun}. More precisely, to draw conclusions as to~$g_{m n}^\rho$ in eq.~\eqref{eq:gstheta}, the problem at stake is to determine the sign of the rational~$\rho = \theta/m$ according to the respective values 
of integers~$k,\, \ell,\, r,\, s\ge 1$. The general form of each striped Riordan matrix in each case can then be deduced from each associated pair~$\b(g_{m n}^\rho , s_{mn}\b)$ in eq.~\eqref{eq:gstheta}.

First, without loss of generality, we can assume~$k\ne \ell$, otherwise ($m = 0$), $g_{m n}^\rho(x) = 1$ 
and~$s_{mn}(x)= x$ (see eq.~\eqref{eq:slambdaLie} in~\S\ref{subst2}), with the identity~$(1,x)$ as associated striped Riordan matrix. Next, the case of~$\ell > k$ ($m > 0$) gives rise to three subcases where, as the case may be, each one yields the expression of~$g_{m n}^\rho(x)$, where~$mn = (\ell - k)(k + \ell) > 0$. So, according to the sign of~$\rho$ only, 
\be
\item If~$s/r = k/\ell$, then~$\rho = \theta = 0$\ and~$g_{m n}^\rho(x) = 1$. As a consequence, 
\begin{equation} \label{eq:thetanul} 
\b(g_{m n}^\rho(x) , s_{mn}(x)\b) =  \l(1 , \frac{x}{(1 - mn\lambda x^n)^{1/n}}\r).
\end{equation}
The corresponding striped Riordan matrices take the general form~$L_{mn,0} = (1 , s_{mn})$.

\item If~$s/r > k/\ell$, then~$\rho > 0$\ and the prefunction simplifies to~$g_{m n}^{\rho}(x) %
= \frac{1}{\b(1 - mn\lambda x^n\b)^{\rho/n}}$. As a consequence,
\begin{equation} \label{eq:thetapos} 
\b(g_{m n}^\rho(x) ,s_{mn}(x)\b) ={} \l(\frac{1}{(1 - mn\lambda x^n)^{\rho/n}} , \frac{x}{(1 - mn\lambda x^n)^{1/n}}\r).
\end{equation}
The corresponding striped matrices take the general form~$L_{mn,\rho} = \b(g_{m n}^\rho , s_{mn}\b)$.

\item If~$s/r < k/\ell$, then~$\rho < 0$\ and the prefunction simplifies to~$g_{m n}^\rho(x) %
= \l(1 - m n \lambda x^n\r)^{|\rho|/n}$. As a consequence,
\begin{equation} \label{eq:thetaneg} 
\b(g_{m n}^{-|\rho|}(x) , s_{mn}(x)\b) ={} \l( (1 - mn\lambda x^n)^{|\rho|/n} , \frac{x}{(1 - mn\lambda x^n)^{1/n}} \r),
\end{equation}
and the associated striped Riordan matrices have the general form~$L_{mn , -|\rho|} = \b(g_{m n}^{-|\rho|},s_{mn}\b)$. 
\ee
Finally, the symmetric case of~$\ell < k$ determines quite a similar behaviour of~$\rho$, except (up to signs) for the respective expressions of~$g_{m n}^{\;\rho}$ in eqns.~\eqref{eq:thetanul},\eqref{eq:thetapos},\eqref{eq:thetaneg}. More precisely, if $\ell < k$, then~$m = \ell - k < 0$, whence the respective forms of~$g_{m n}^\rho$:
\begin{flalign} \label{eq:glkthetapos}
g_{m n}^\rho(x) &={} 1\ \quad \text{if}\ \rho = 0,\nonumber\\
g_{m n}^\rho(x) &={} \frac{1}{\b(1 + |m| n \lambda x^n\b)^{\rho/n}}\ \;\qquad \tif\ \rho > 0 \qquad \tand\\
g_{m n}^\rho(x) &={} \l(1 + |m| n \lambda x^n\r)^{|\rho|/n}\ \qquad \tif\ \rho < 0.\nonumber
\end{flalign}
Therefore, the three above expressions of~$g_{m n}^{\;\rho}$ in~\eqref{eq:glkthetapos} produce three classes of striped Riordan matrices of general forms~$\b(1,\ov{s_{mn}}\b)$, $\b(g_{m n}^{\rho} , \ov{s_{mn}}\b)$
and~$\b(1/g_{m n}^{|\rho|} , \ov{s_{mn}}\b)$. Altogether, the prefunction admits the general expression
\[
g_{m n}^\rho(x) ={} \frac{1}{\b(1 \pm |m| n \lambda x^n\b)^{\pm |\rho|/n}}.\]

\begin{ex} \label{ex:r=s}
Again, if~$r = s$, then~$\theta = rm$ and~$\rho = r\in \Z_{>0}$. Each one of the solutions simplifies to the following statements.
\bi
\item If~$k = \ell$, $m = 0$ and then~$\b(g_{m n}^r(x) , s_{mn}(x)\b) = (1,x)$, i.e. the identity~$I\in \cG(n,r)$.

\item If~$k < \ell$, then~$\b(g_{m n}^r(x) , s_{mn}(x)\b) = \delimiterfactor=1500 
\b(\tfrac{1}{(1 - m n \lambda x^n)^{r/n}} , \tfrac{x}{(1 - m n \lambda x^n)^{1/n}}\b)$, which is an element of the form~$L_{mn,\rho} = \b(g_{m n}^{r} , s_{m n}\b)$ in~$\cG(mn,r)$\ $(m>0)$.

\item If~$k > \ell$, then~$\b(g_{mn}^{\rho}(x) , s_{mn}(x)\b) = \delimiterfactor=1500
\b(\tfrac{1}{(1 + |m| n \lambda x^n)^{r/n}} , \tfrac{x}{(1 + |m| n \lambda x^n)^{1/n}}\b)$, which is an element of the form~$L_{mn,r}^{-1} = \b(g_{-|m|n}^{\,r} , s_{-|m|n}\b)$ in~$\cG(mn,r)$\ $(m<0)$.
\ei
One of the simplest applications is when~$r = 1$ and~$m = \ell - k = 1$, which correspond to the Pascal 
matrix~$P = \b(\frac{1}{1-x},\frac{x}{1-x}\b)$ and its inverse~$P^{-1} = \b(\frac{1}{1+x},\frac{x}{1+x}\b)$ (see Ex.~\ref{ex:binmatstirling}).
\end{ex}
To summarize, whatever~$\ell,\, k\in \Z$ and~$\rho\in \Q$, the discussion leads to the two matrices~$L_{mn,\rho}$\ and\ $L_{mn,\rho}^{-1}$ (where~$m := \ell - k$, $n := k + \ell$\ and~$\rho := \theta/m$), which are the two generators 
of~$\cG(mn,\rho)$.

\begin{rem} \label{rem:genliebragnrho}
Turning back to Remark~\ref{rem:genliebra} of~\S\ref{Liebprefun} regarding differential operators of the 
form~$x^k\vthx - rx^k$ or~$x^\ell\vthx - sx^\ell$, we obtain a general expression of the prefunction~$g_{m n}^\rho(x)$ as follows (still on the assumption that~$k\ne \ell$, otherwise~$m = 0$\ so~$g_{m n}^\rho(x) = 1$ for any~$\rho\in \Q$).
\be
\item[$(i)$] In the case when the scalar part of the Lie bracket writes~$\pm |\theta| x^n$, 
where again~$\theta := s\ell - rk$ and~$\rho = \theta/m$, the expressions of the prefunction~$g_{m n}^\rho$ are similar to the above cases in eqs.~\eqref{eq:thetanul},\eqref{eq:thetapos},\eqref{eq:thetaneg},\eqref{eq:glkthetapos}, according to the sign of~$\rho m$.

\item[$(ii)$] In the case when the scalar part of the Lie bracket is~$-(rk + s\ell) < 0$, the sign 
of~$\tilde{\rho} := -\frac{rk+s\ell}{m}$ depends only on whether~$m := \ell - k > 0$ or not, which implies the expression of the prefunction,
\[
g_{m n}^\rho(x) ={} 
\begin{dcases*} 
\l(1 - m n \lambda x^{n}\r)^{-|\tilde{\rho}|/n} & \tif\ \ $k < \ell$\ \qquad \text{or}\\ 
\l(1 + |m| n \lambda x^{n}\r)^{-\tilde{\rho}/n} & \tif\ \ $k > \ell$.
\end{dcases*}\] 
\ee
\end{rem}

\section{Striped quasigroup and semigroup operations} \label{extlaw}
The subgroup~$\cG(n,\rho)$ of~$\cG(\rho)$ is introduced in Def.~\ref{def2} and eq.~\eqref{lnrho} in the latter Section~\ref{groupsrlm}. We consider now the set~$\cG(n,\rho;\mu)$ with~$\mu\in \Q$ of all 
pairs~$L_{\mu n,\rho} := L_{n,\rho}^{(\mu)} = \l(g_{\mu n}^\rho , xg_{\mu n}\r)$, where~$g_{\mu n}\in \cS$ is of the 
form~$\b(1 - \mu n \lambda x^n\b)^{-1/n}$. $\cG(n,\rho;\mu)$ is constructed from the group~$\b(\cG(n,\rho),*\b)$ and endowed with a new binary operation defined from the results obtained in~\S\ref{lbgroup}.

\begin{rem} \label{}
Let~$n\in \Z_{>0}$. For~$\mu\in \R$, any word of the form~$\omega = \mu a^{+^{(n+1)}} a$ 
in~$\rHW_{\le 1}\!\setminus \!\rHW_0$ spans a vector space over~$\R$. For~$\mu$ restricted to~$\Q$, $\omega$ spans a $\Q$-submodule of~$\rHW_{\le 1}\!\setminus \!\rHW_0$.\par
Observe that, by taking~$\mu \in \Z$ (denoted~$m$ in that case: see Notations~\ref{nota}), the substitution function~$s_n(x)\in x\cS$ corresponding to~$\omega$ verifies by definition, with~$s_{mn} := s_{n}^{(m)}$,
\[
s_{mn} := s_n\circ \cdots \circ s_n\ \ \text{if}\ m\in \Z_{\ge 0}\ \quad \text{and}\ \quad %
s_{-mn} := s_n^{(-1)}\circ \cdots \circ s_n^{(-1)}\ \ \text{if}\ m\in \Z_{\le 0}.\]
That is, for all~$m\in \Z$, $s_{\pm mn}(x)=  x\b(1 \pm m n\lambda x^n\b)^{-1/n}$, as the case may be. 
Since~$s_{mn}\circ s_{-mn} = x$, each~$s_{mn}$ and~$s_{-mn}$ is the reverse series of the other in~$x\cS$. Moreover, this enables to interpolate between~$g_{mn}\in \cS$ and~$s_{mn}\in x\cS$ for any~$m \in \Z$. From the Lie group structure of the interpolation unipotent subgroup of~$\b(\cG(n,\rho),\circ\b)$, $m$ can be taken in $\Q$, $\R$ 
or~$\C$ (see Bacher in~\cite{Bacher06}). So, in order to avoid any confusion, the parameter~$\mu\in \Q$ will henceforth replace the parameter~$m$ previously used in Section~\ref{groupsrlm}.
\end{rem}

\subsection{The striped quasigroup for the operation $\ov{*}$} \label{quasigrp} 
\begin{defi} \label{def:squasigrp}
Let~$\mu\in \Q$ denote the~$\mu$-th compositional power of the matrices which belong to the Riordan 
subgroup~$\b(\cG(n,\rho),*\b)$ for any~$n\in \Z_{>0}$ and $\rho\in\Q$. The set under consideration is
\[ \delimiterfactor=1200
\cG(n,\rho;\mu) = \l\{L_{\mu n,\rho} := L_{n,\rho}^{(\mu)}\, \mid \,L_{n,\rho}\in \cG(n,\rho),\ n\in \Z_{>0}\ %
\text{and}\ \mu, \rho\in \Q \r\},\]
equipped with the binary operation~$\ov{*}$. In accordance with the discussion on the prefunction in~\S\ref{lbgroup},
$g_{\mu n}^\rho (x) = \b(1 - \mu n \lambda x^n\b)^{-\rho/n}$ supplied by the Lie bracket related to the 
group~$\cG(n,\rho)$, the product~$\ov{*}$ is defined as follows.\par
Let~$k, \ell\in \Z_{>0}$, $r, s$ and~$\sigma, \tau$ in $\Q$. 
Then, for any pair of matrices~$\b(L_{k,r} , L_{\ell,s}\b)\in \cG(k,r)\times \cG(\ell,s)$ (see eq.~\eqref{lnrho}, 
\S\ref{groupsrlm} and~\S\ref{lbgroup}),  
\[
L_{\sigma k,r} \;\ov{*}\; L_{\tau \ell,s} = L_{\mu n,\rho}\in \cG(n,\rho;\mu),\] 
where~$n = k + \ell\in \Z_{>0}$, and~$\rho = \varphi(r,s) = (s\ell - rk)/m$\ and~$\mu = \sigma \tau m$\ 
(with~$m = \ell - k\in \Z$) are in~$\Q$.
\end{defi}

\begin{rem} 
Given any pair~$(n,\rho)\in \Z_{>0}\times \Q$, the product~$\ov{*}$ operates on the two generators~$L_{n,\rho}$ 
and~$L_{n,\rho}^{-1}$ of~$\cG(n,\rho)$ as in Def.~\ref{def:squasigrp}. For any~$\mu\in \Q$, 
$L_{-\mu n,\rho} := L_{n,\rho}^{(-\mu)}$ is the inverse of~$L_{\mu n,\rho}:= L_{n,\rho}^{(\mu)}$ and thus, 
$L_{\mu n,\rho}\; \ov{*}\; L_{-\mu n,\rho} = (1,x)$, the identity in~$\cG(n,\rho)$. Note that, in the particular case when~$r = s$\ ($\varphi(r,r) = r$), one can find again in Def.~\ref{def:squasigrp} the rational value~$\rho = \theta/m$ in~\S\ref{lbgroup} with~$\theta := s\ell - rk\in \Q$.
\end{rem}
The next Def.~\ref{def:weakasso} of \emph{weak (or non strict) associativity} appears primarily important for characterizing the structure of the set~$\cG(n,\rho;\mu)$ and the operations on its elements in this section.

\begin{defi} \label{def:weakasso}
The operation~$\ov{*}$ will be said {\em weakly associative} iff it is not associative, {\em unless} it is closed under belonging to the conjugacy classes in~$\cG(n,\rho;\mu)$.\par
Namely, let~$j,\, k,\, \ell\in \Z_{>0}$ and~$r,\, s,\, t\in \Q$. Given~$\mu_1, \mu_2, \mu_3\in \Q$, the 
operation~$\ov{*}$ is said {\em weekly associative} iff there exist at least three matrices~$L_{j,r} \in \cG(j,r)$, 
$L_{k,s} \in \cG(k,s)$ and $L_{\ell,t} \in \cG(\ell,t)$ meeting the three conditions
\begin{flalign*}
&(i)\ L_{\mu_1j,r}\; \bar{*}\; \b(L_{\mu_2k,s} \;\bar{*}\; L_{\mu_3\ell,t}\b) \ne %
\b(L_{\mu_1j,r} \;\bar{*}\; L_{\mu_2k,s}\b) \; \bar{*}\; L_{\mu_3\ell,t}; &\\
&\delimiterfactor=1100 %
(ii)\, L_{j,r}^{(\mu_1)} \,\bar{*}\, \b(L_{k,s}^{(\mu_2)} \bar{*} L_{\ell,t}^{(\mu_3)}\b) \in %
\cG_{j+k+\ell,\varphi(r,\varphi(s,t))}\ \,\text{and}\, \delimiterfactor=1100 %
\b(L_{j,r}^{(\mu_1)} \bar{*} L_{k,s}^{(\mu_2)}\b) \,\bar{*}\, L_{\ell,t}^{(\mu_3)} \in %
\cG_{j+k+\ell,\varphi(\varphi(r,s),t)}; &\\
&(iii)\ \varphi\b(r,\varphi(s,t)\b) \ne \varphi\b(\varphi(r,s),t\b). &
\end{flalign*}
\end{defi}
By contrast, if condition~$(iii)$ is not fulfilled whereas~$(i)$ and~$(ii)$ still hold, there exist~$r,\, s,\, t\in \Q$ such that~$\varphi\b(r,\varphi(s,t)\b) = \varphi\b(\varphi(r,s),t\b)$ (e.g., if~$r = s = t$ 
then~$\varphi\b(s,\varphi(s,s)\b) = \varphi\b(\varphi(s,s),s\b) = s$). Hence, the product~$\ov{*}$ is associative by closure within one unique reduced conjugacy class in the set~$\cG(n,\rho;\mu)$, whenever~$n = k + \ell + m$ and, say 
$\rho = \varphi\b(r,\varphi(s,t)\b)$, for instance.

In other words, although the operation~$\ov{*}$ is never associative (in the strict sense), a weak form of associativity holds however in~$\cG(n,\rho;\mu)$, according to the cases. On the one hand (in the general case), the product~$\ov{*}$ makes the set of all the conjugacy classes in the Riordan subgroup~$\cG(n,\rho)$ closed under weak associativity. On the other hand, whenever associativity is proved valid within at least one conjugacy class in the group~$\cG(n,\rho)$ for every pair~$(n,\rho)\in \Z_{>0}\times \Q$ (e.g., for say~$\rho = \varphi\b(r,\varphi(s,t)\b)$, one unique conjugacy class is closed under the product~$\ov{*}$).

\begin{theorem} \label{quasigrpstruct}
With respect to the operation~$\ov{*}$, $\cG(n,\rho;\mu)$ has a \emph{quasigroup} structure.
\end{theorem}
\begin{proof}
The set~$\cG(n,\rho;\mu)$\ in Def.~\ref{def:squasigrp} is closed under the binary operation~$\ov{*}$. However, 
$\cG(n,\rho;\mu)$\ has no identity element (except for~$\rho = \mu = 0$) and neither exists an inverse for every element~$L_{\mu n,\rho}$\ indexed by~$n\in \Z_{>0}$ and $\rho,\, \mu\in \Q$. Therefore, 
by Def.~\ref{def:squasigrp} and Def.~\ref{def:weakasso}, the product~$\ov{*}$ is not associative. 
The latter properties make~$\b(\cG(n,\rho;\mu),\ov{*}\b)$ into a \emph{quasigroup}\footnote{In 1914, \'Elie Cartan first proved that the smallest exceptional Lie groups~$s_2$ is the automorphism group of the {\em Octonions}, which examplifies the  structure of quasigroup. The quasigroups have been investigated further, e.g. in~\cite{Adams96}, and then explored by Baez, Convey or Smith and Smith~\cite{Smith07}, among others.}.
\end{proof}
Moreover, the following Lemma~\ref{anticomm} shows that the product~$\ov{*}$ is anticommutative in the case 
when~$\rho = 0$ and~$\mu\ne 0$.

\begin{lem} \label{anticomm}
The quasigroup~$\cG(n,\rho;\mu)$ is endowed with a non-commutative operation~$\ov{*}$ for all~$n\in \Z_{>0}$, $\mu$ 
and~$\rho\in \Q$. Moreover, when~$\rho = 0$ and~$\mu\ne 0$, the operation is also anticommutative.
\end{lem}
\begin{proof}
The operation~$\ov{*}$ is obviously non-commutative for~$\cG(n,\rho;\mu)$, whatever the values~$n\in \Z_{>0}$ 
and~$\mu,\, \rho\in \Q$. Now, let~$\rho = 0$ and~$\mu\ne 0$. By Def.~\ref{def:squasigrp}, for any two matrices~$L_{k,r}\in \cG(k,r)$ and~$L_{\ell,s}\in \cG(\ell,s)$,
\[
L_{\sigma k,r} \;\ov{*}\; L_{\tau \ell,s} = \l(1,xg_{\mu n}\r) \ne %
L_{\tau \ell,s} \;\ov{*}\; L_{\sigma k,r} = \l(1,xg_{\mu n}\r)^{-1},\] 
where~$g_{\mu n} = \l(1 - \mu \lambda x^n\r)^{-1/n}$, for all~$n = k + \ell\in \Z_{>0}$ 
and~$\mu = \sigma \tau(\ell - k)\in \Q^*$. Hence, if~$\rho = 0$ the operation~$\ov{*}$ is anticommutative for the quasigroup~$\cG(n,0;\mu)\!\setminus \!\b\{(1,x)\b\}$, and the lemma follows.
\end{proof}

\begin{rem} \label{rem:inter}
The quasigroup~$\cG(n,0;\mu)$ is a ``prefunction~$n$-striped subset'' of the Lagrange Riordan subgroup 
$\cL = \b\{(1,xg_{\mu n})\b\}$ ($g_{\mu n}\in \cS$, $\mu\ne 0$) under operation~~$\ov{*}$.

If~$\mu = 0$ or~$\mu = 1$, then for any~$n\in \Z_{>0}$ and $\rho\in\Q$ the quasigroup~$\cG(n,\rho;\mu)$ reduces, respectively, either to the group~$\{(1,x)\}$, or to the group~$\cG(n,\rho; 1)\simeq \cG(n,\rho)$ under the 
product~$\ov{*}$. Moreover, if~$\mu = 0$, we have~$\dis \bigcap_{n\in \Z_{>0}} \cG(n,\rho;0) = \b\{(1,x)\b\}$, since~$(1,x)$ is the unique element that belongs to all sets~$\cG(n,\rho)$\ for any~$n\in \Z_{>0}$\ and $\rho\in\Q$.
\end{rem}

\subsection{The striped semigroup for the operation $\cirast$} \label{semigrp}
\begin{defi} \label{def:ssemigrp}
Under the same assumptions as in Def.~\ref{def:squasigrp}, the set~$\cG(n,\rho;\mu)$ can be extended to the union set 
\[\delimiterfactor=1200 \cH(n,\rho) := \Union{\mu\in \Q} \cG(n,\rho;\mu)\ \ \text{with}\ \ n\in \Z_{>0}\ %
\text{and}\ \rho\in \Q,\] 
where $\cH(n,\rho)$ can be endowed with a binary operation~$\cirast$. 
Considering again the prefunction $g_{\mu n}^\rho(x) = \b(1 - \mu n \lambda x^n\b)^{-\rho/n}$ supplied by the Lie bracket given in Def.~\ref{def:squasigrp} (\S\ref{quasigrp}), the product~$\cirast$ is defined as follows for 
any~$k,\, \ell\in \Z_{>0}$\ and any rational~$r$, $s$, $\sigma$ and~$\tau$:
\[
\cG(k,r;\sigma)\, \cirast\, \cG(\ell,s;\tau) = \cG(n,\rho;\mu),\]
where we still let~$n := k + \ell$, $\rho := \varphi(r,s) = (s\ell - rk)/m$\ and~$\mu := \sigma \tau m$ 
with~$m := \ell - k$.
\end{defi}

\begin{lem} \label{involution}
Let~$k,\,\ell\in \Z_{>0}$\ and~$r$, $s$, $\sigma$, $\tau\in \Q$. If~$\mu = \sigma \tau m = 0$, 
then~$\cG(k,r;\sigma)\, \cirast\, \cG(\ell,s;\tau) = \{(1,x)\}$, which is the identity in~$\cG(n,\rho)$. Therefore, the binary operation~$\cirast$ is an involution within~$\cH(n,\rho)$ either if~$n = 2k$ ($k=\ell$), or else 
($k\neq \ell$) provided~$\sigma \tau = 0$.
\end{lem}
\begin{proof}
Given~$n = k + \ell\ge 2$ and~$m = \ell - k\in \Z$, we can generate all the elements~$\cG(n,\rho;\mu)$ 
in~$\cH(n,\rho)$ by composition of~$\cG(n,r;\sigma)$ with~$\cG(\ell,s;\tau)$\ for all~$\sigma,\, \tau\in \Q$. Indeed, we can always find two rationals~$\sigma$ and~$\tau$ such that there exists~$\mu = \sigma \tau m$ in~$\Q$ with~$L_{\mu n,\rho} = L_{\sigma k,r}\,\ov{*}\, L_{\tau \ell,s}$ for any matrix~$L_{\mu n,\rho}$. Hence, $\cirast$ is an involution within~$\cH(n,\rho)$ if, and only if, $\mu = \sigma \tau m = 0$, that is either if~$m = 0$\ (i.e. $k = \ell$ and~$n = 2k$), or~$m\neq 0$\ and~$\sigma = 0$ or~$\tau = 0$. This completes the proof of the lemma.
\end{proof}

\begin{theorem} \label{semigrpstruct}
With respect to the operation~$\cirast$, $\cH(n,\rho)$ has a \emph{semigroup} structure.
\end{theorem}
\begin{proof}
The set~$\cH(n,\rho)$ is closed under the binary operation~$\cirast$. However, by lemma~\ref{involution}, 
$\cH(n,\rho)$ has no identity element for all values of~$\mu\in \Q$, except when~$\mu = 0$. Furthermore, there exists no inverse for every element~$\cG(n,\rho;\mu)$ indexed by~$n\in \Z_{>0}$ and~$\rho,\,\mu \in \Q$. Yet, 
$\cH(n,\rho)$\ may be embedded into a monoid simply by adjoining the set~$\{(1,x)\}$, which is not a subset of the set~$\cH(n,\rho)$. Therefore, the identity may be defined in~$\cH(n,\rho)\cup \{(1,x)\}$.\par
There remains to prove that the product~$\cirast$ is indeed associative. Let~$j,\, k,\, \ell\in \Z_{>0}$ 
and $r,\, s,\, t\in \Q$. By Def.~\ref{def:weakasso}, for all~$\mu_1,\, \mu_2,\, \mu_3\in \Q$, the operation~$\cirast$ meets the property that
\[ \delimiterfactor=1200
\cG(j,r;\mu_1) \;\cirast\; \l(\cG(k,s;\mu_2) \,\cirast\, \cG(\ell,t;\mu_3)\r) %
= \l(\cG(j,r;\mu_1) \,\cirast\, \cG(k,s;\mu_2)\r) \;\cirast\; \cG(\ell,t;\mu_3).\]
Hence, $\b(\cH(n,\rho),\cirast\b)$ is a semigroup and the theorem is established.
\end{proof}

\begin{rem}
A monoid is a semigroup with an identity element. Any semigroup~$S$ may be embedded into a monoid (generally denoted 
as~$S^{1}$) simply by adjoining an element~$e$ not in~$S$ and defining $es = s = se$ for all $s\in S\cup \{e\}$. In the present case, the set $\{(1,x)\}$ plays the role of the identity in the monoid~$\cH(n,\rho)\cup \{(1,x)\}$\ denoted here as~$\cH(n,\rho)^{1}$.
\end{rem}

\addcontentsline{toc}{section}{References}
\section*{References}

\vskip -1cm
\def\refname{\empty}

\bibliographystyle{article}
\def\bibfmta#1#2#3#4{ {\sc #1}, {#2}, \emph{#3}, #4.}
\bibliographystyle{book}
\def\bibfmtb#1#2#3#4{ {\sc #1}, \emph{#2}, {#3}, #4.}

\newpage
\pagenumbering{roman}

\appendix

\addcontentsline{toc}{section}{Appendix}
\section*{Appendix}
\renewcommand{\thesubsection}{\Alph{subsection}}
\numberwithin{equation}{subsection}

\subsection{Heisenberg Lie algebra~$\cL_{\cH}$ and the enveloping algebra~$\cU\l(\cL_{\cH}\r)$} \label{app:hla}

\emph{Heisenberg Lie algebra}, denoted by~$\cL_{\cH}$, is a 3-dimensional vector space with 
basis~$\{a^+,a,c\}$ and Lie bracket defined by~$[a,a^+] = c$,~$[a^+c] = [ac] = 0$, where~$c$ is the centre of~$\cL_{\cH}$. Passing to the enveloping algebra involves imposing the linear order~$a^+\prec a\prec c$ and constructing the \emph{enveloping algebra}~$\cU\l(\cL_{\cH}\r)$ with basis given by~$\l\{(a^{+})^k a^\ell c^m\r\}$, which is indexed by triples of non-negative integers~$k, \ell, m$. Hence, the elements in~$\cU\l(\cL_{\cH}\r)$ are of the form
\[
\sum_{k,\ell,m} \beta_{k,\ell, m} (a^+)^k a^\ell c^m.\]
According to the theorem of Poincar\'e--Birkhoff--Wick (see~\cite{DPSPBH10}), the associative product law 
in~$\cU\l(\cL_{\cH}\r)$ is defined by concatenation, subject to the rewriting rules
\[
aa^+ = a^+a + c,\ \quad ca^+ = a^+c\ \quad \text{and}\ \ ca = ac.\]
The formula for multiplication of basis elements in~$\cU\l(\cL_{\cH}\r)$ is thus a slight generalization of 
eq.~\eqref{eq:sc} in lemma~\ref{strconst}, where~$k, \ell, p, q, r, s\in \N$,
\[
(a^{+})^k a^\ell c^p (a^{+})^r a^s c^q = \sum_{i=0}^{\min(\ell,r)} i!\, \bin{\ell}{i}\bin{r}{i}\, %
(a^{+})^{k+r-i} a^{\ell+s-i} c^{p+q+i}.\]
The enveloping algebra~$\cU\l(\cL_{\cH}\r)$ differs from~$\cH$ by the additional central element~$c$ which should not be confused with the unity of the enveloping algebra, denoted by~$1$. When this difference is insubstantial, one may set~$c\map 1$, thus recovering~$\rHW_{\C}$, the algebra of Heisenberg--Weyl~$\cH$, i.e. we have the surjective 
morphism~$\kappa : \cU\l(\cL_{\cH}\r)\lra \cH$ given by~$\kappa\b(a^{+^i} a^j c^k\b) = a^{+^i} a^j$.

Furthermore, grading~$\rHW_{\C}$ can be carried out as follows. Let~$E\in \Z$ be the excess of~$b_{i,j} = (a^+)^i a^j$ (a natural linear basis of~$\rHW_{\C}$). Setting~$\rHW_\C^{(E)} = \Span_\C (b_{i,j})_{i-j=E}$, one has
\[
\rHW_\C = \bigoplus_{E\in \Z} \rHW_\C^{(E)}\ \quad \tand \quad %
\rHW_\C^{(E_1)}\, \rHW_\C^{(E_2)}\subset \rHW_\C^{(E_1+E_2)},\]
for all~$E_1,\, E_2\in \Z$. This natural grading makes~$\rHW_{\C}$ into a~$\Z$-graded algebra~\cite{DuPT11,DPSPBH10} through the (faithful) representation of Bargmann--Fock defined in \S\ref{bfr} (Section~\ref{diffops}). Indeed, 
the regradation of~$\C\langle A,B\rangle$ given by~$\deg(B) = -\deg(A) = 1$, makes~$\cI_{\rHW}$ into a graded ideal. (This gradation appears in fact natural when the algebra~$\rHW_{\C}$ is represented by the two usual differential operators~$X$ and~$D$.)

\subsection{Evaluation of the substitution function in a pure vector field} \label{app:intopg}
From Bargmann--Fock representation defined in~\S\ref{bfr} (Section~\ref{diffops}), the basic formula~\eqref{eq:basicdop} is considered here when the problem involves only a pure vector field. Then, the transformation 
of~$f(x)\ne 0$ is given by one substitution only and eq.~\eqref{eq:basicdop} simplifies to the basic formula
\begin{equation} \label{appeq:basic}
e^{\lambda q(x) \partx} \b[f(x)\b] = f\b(s_\lambda(x)\b),
\end{equation} 
where~$q$ is at least continuous,~$\lambda\in \R$ is sufficiently small and~$f$ is an holomorphic function 
on~$(0,1)$ (see Comments~\ref{comm:shift}, Section~\ref{onepargroup}).

Following~\cite{DOTV97,DPSHB04,DuPT11}, we choose an open interval~$I\ne \emptyset$ where~$q(x) \not\equiv 0$ 
and~$x_0 \in I$. For any~$x\in I$, we set
\[
F(x) := \int_{x_0}^{x} \frac{\rmd t}{q(t)}\,.\]
Now, let~$J = F(I)$ be an open interval. The function~$F : I \lra J$ is strictly monotone, so it is a diffeomorphism. We put~$s_\lambda(x) = F^{-1} \b(F(x) + \lambda\b)$ for any pair~$(x,\lambda)$ for which the above expression 
of~$s_\lambda(x)$ makes sense, that is on the domain 
\[
\cD_{I,q} := \l\{(x,\lambda) \in \R^2\, \mid x\in I,\, (F(x)+\lambda) \in J\r\}.\]
Since~$(x,\lambda)\map s_\lambda(x)$ is continuous (and even of class~$\cC^1$ upon its domain and~$s_0(x) = x$, 
$s_\lambda$ is as a deformation of the identity. For small values of~$\lambda$, the operator~$e^{\lambda q(x)\partx}$ coincides with the substitution factor~$f\map f\circ s_\lambda$, for the exponential of a derivation (such 
as~$\lambda q(x)\partx$) is an automorphism, i.e. a substitution in the (test) function spaces under consideration.

Notice again that the validity of the formula~$e^{\lambda q(x)\partx} \b[f(x)\b] = f\b(s_\lambda(x)\b)$ for small parameter~$\lambda$ can be restricted due to the nature of the function~$f$ itself. As a matter of fact, the discussion carried out in Comment~\ref{comm:shift} (Section~\ref{onepargroup}) shows that~$f$ has to 
be~$\cC^\omega\subset \cC^\infty$, i.e. analytic on the open set~$(0,1)$ (in the present context), in order to fulfil the shift preserving requirements.

In the case when elements have the form~$a^{+^{n}} a$ with~$n\ge 2$ (see \S\ref{substpref} 
in Section~\ref{onepargroup}), the operator takes the form~$e^{\lambda x^n \partx}$. So, the diffeomorphism 
$F : (1,+\infty) \lra \b(-\infty, \tfrac{1}{n-1}\b)$ is defined as~$F(x) = \int_1^{x} \frac{\rmd t}{t^n} %
= \frac{1-x^{-(n-1)}}{n-1}$ and, from Remark~1, the substitution factor is (with integer~$n\ge 2$)
\[ 
s_{\lambda}(x) = \frac{x}{\sqrt[\uproot{2} n-1] 1 - (n-1) \lambda x^{n-1}},\ \qquad %
(\text{for}\ \ |x| < \b((n-1)\lambda\b)^{-1/(n-1)}).\]

\no \textit{Remark}\ (Link with local Lie groups: Straightening the vector fields on the line)

Starting from the above operator~$e^{\lambda x^n \partx}$ defined on the whole line, it is possible (at least locally) to straighten this vector field by a diffeomorphism~$u$ to get the constant vector field. As the one-parameter group generated by a constant field is a shift, the one-parameter (local) group of transformations will be, on a suitable domain
\[
U_\lambda[f](x) = f\l(u^{-1} (u(x) + \lambda)\r).\]
Now, it is known that, if two one-parameter groups have the same tangent vector at the origin, then they coincide. Direct computation gives this tangent vector
\[
\l. \frac{\rmd}{\rmd\lambda}\r|_{\lambda=0} f\l(u^{-1} (u(x) + \lambda)\r) = \frac{1}{u'(x)} f'(x).\]
So the local one-parameter group~$U_\lambda$ admits~$1/u'(x) \partx$ as tangent vector field. Here, we have 
to solve~$1/u'(x) = x^n$ (i.e. $u(x) = \int_1^{x} \frac{\rmd t}{t^n}$) in order to get the diffeomorphism~$u$. The above calculations, already performed for~$F(x)$, yield the substitution factor~$s_\lambda$.

\m \no \textit{Example}\ In the three cases briefly mentioned in subsection~\ref{intopg}, $s_\lambda$ is similarly evaluated.

If~$n = 0$, all elements reduce to an annihilation. The operator takes the form~$e^{\lambda \partx}$, 
so we get~$F(x) = \int_{0}^x \rmd t = x$, $s_\lambda(x) = F^{-1} \b(F(x) + \lambda\b) = x + \lambda$ 
and~$e^{\lambda \partx} [f(x)] = f\l(s_\lambda(x)\r) = f(x + \lambda)$. From a geometric viewpoint the corresponding substitution~$s_\lambda(x)$ corresponds to a translation.
 
If~$n = 1$, all elements are in the form~$a^+ a$. The operator is given by~$e^{\lambda \vthx}$ and so we have~$F(x) = \int_{1}^x \frac{\rmd t}{t} = \ln(x)$,~$s_\lambda(x) = \exp\b(\ln(x)+\lambda\b) = xe^\lambda$ and~$e^{\lambda \vthx} [f(x)] = f\l(xe^\lambda\r)$. The geometric transformation associated to~$s_\lambda$ is a homothety (dilation or contraction).

If~$n = 2$, all elements have the form~$a^{+^2} a$, the operator turns out to be~$e^{\lambda x\vthx}[f(x)] %
= f\l(\frac{x}{1 - \lambda x}\r)$, the diffeomorphism~$F : (0, +\infty) \lra (-\infty, 1)$ 
is~$F(x) = \int_{1}^{x} \frac{\rmd t}{t^2} = 1 -\frac{1}{x}$ and the substitution factor 
is~$s_\lambda (x) = \frac{x}{1 - \lambda x}$, whose geometric transformation corresponds to a homography.

Notice also that, for such small values of~$n$, finding the substitution function may prove much easier through a direct use of operational calculations on monomials (and thus polynomials)~\cite{DOTV97,Roman05}. So, considering 
the operator~$e^{\lambda \partx}$ on~$f\in \C[x]$ through its symbolic Taylor's expansion in~$\lambda$ yields\ 
$e^{\lambda \partx} \l[f(x)\r] = \sum_{k} \frac{\lambda^k}{k!} \partx^k f(x) %
= \sum_{k=0}^n \bin{n}{k} \lambda^k f^{(n-k)}(x) = f(x + \lambda)$. Next, under the change of variable~$x = e^{\vthx}$, 
$e^{\lambda x\partx} [x^n] = e^{\lambda \partial_{\vthx}} f\l(e^{\vthx}\r) = f\l(e^{\lambda+\vthx}\r) = %
f\l(x e^\lambda\r)$. Similarly, under the change of variable~$x = 1/y$, one gets~$e^{\lambda x^2\partx} [f(x)] = %
\delimiterfactor=800 f\l(\tfrac{x}{1-\lambda x}\r)$. (The generalization to any operator of the 
form~$e^{\lambda x^n \partx}$ can thus be found by induction on~$n$ under appropriate changes of variable.) 

\subsection{Formal power series, Laurent and Puiseux formal power series} \label{app:fpsgf}
Let~$\K$ be a ring ($\Z$,~$\Q$,~$\R$,~$\C$). Formal power series extend the usual algebraic operations on polynomials with coefficients in~$\K$ to infinite series of the form~$f(x) = \sum_{n\ge 0} f_n z^n$, where~$z$ is a formal indeterminate.~$\K[[z]]$ denotes the ring of power series on~$\K$.~$\K[[z]]$ is the set~$\K^\N$ of infinite sequences of elements of~$\K$, written as infinite sum~$f(z)$, endowed with the operations of sum and Cauchy (or convolution) product. Similarly, in the case when~$\K$ is a commutative field, the~$\K$-vector space of such infinite sequences indexed by~$\N$ is the commutative~$\K$-algebra of power series~$\K[[z]]$ with one indeterminate in~$\K$, and the polynomials form a subalgebra~$\K[z]$ of the algebra~$\K[[z]]$. The composition~$f\circ g$ of two power series is only defined if~$g_0 = g(0) = 0$. Hence, any power series~$f$ such that~$f_0 = f(0) = 0$ has a {\em compositional inverse} (or {\em reverse} series) denoted by~$\ov{f}$, provided that~$f_1 = f'(0)$ is invertible in~$\K$. Any formal series~$f$ such that~$f(0)\ne 0$ (i.e.~$f_0$ is invertible in~$\K$) has a multiplicative inverse denoted by~$f^{-1}$ or~$\tfrac{1}{f}$. 

The {\em order} ord$(f)$ (or valuation) of~$f\in \K[[z]]$ is defined as the smallest integer~$r$ for which the coefficient~$f_r$ of~$z^r$ does not vanish (one sets ord$(0) = +\infty$). The series~$f$ has a {\em multiplicative inverse}, denoted by~$f^{-1}$, if and only if it is invertible in~$\K[[z]]$, i.e. ord$(f) = 0$. The reverse series~$\ov{f}$ of~$f$ satisfies~$f\l(\ov{f}(z)\r) = \ov{f}\b(f(z)\b) = z$ if, and only if, ord$(f)\ge 1$.

A topology, known as the {\em formal topology}, is put on~$\K[[z]]$ by which two series are ``close'' if they coincide to a large number of terms. Given two power series~$f$ and~$g$, their distance~$d(f,g)$ is then defined 
as~$2^{-\text{ord}(f-g)}$. With this metric (an ultrametric distance), $\K[[z]]$ becomes a {\em complete metric space}. The limit of a sequence of series~$\l\{f^{(j)}\r\}$ exists if, for each~$n$, the coefficient of order~$n$ 
in~$\l\{f^{(j)}\r\}$ eventually stabilizes to a fixed value when~$j\to \infty$. In this way, {\em formal convergence} can be defined for infinite sums: it is sufficient that the general term of the sum should tend to~0 in the formal topology (i.e. provided the order of the general term tends to~$\infty$); and similarly for infinite products: 
$\prod \l(1+u^{(j)}\r)$ converges as soon as~$u^{(j)}\to 0$ in the formal topology. Thereby,~$\sum_{k\ge 0} f^k$ exists whenever~$f_0 = 0$, and the definition of the {\em quasi-inverse}~$(1-f)^{-1}$ follows. Formal logarithms, exponentials, derivatives, primitives, etc. may be defined in the same way. When~$\K$ is a field of characteristic 0 ($\C$,~$\R$ or~$\Q$), the Lagrange inversion formula provides a powerful tool to compute the coefficients of~$\ov{f}$ from the coefficients of the powers of~$f$. (See~\cite[App.~A5]{FlSe09}.)

A \emph{Laurent power series} is a formal power series in~$\K[[z]]$ of the form~$\sum_{n\in \Z} a_n z^n$, whose coefficients~$a_n$ with~$n < 0$ are all zero, but for a finite number of them all. Every non null Laurent power series is thus written~$L = a_kz^k + a_{k+1}z^{k+1} + \cdots$, where~$k\in \Z$ and~$a_k\ne 0$. The ring of the Laurent power series is denoted by~$\laurent{\K}{z}$ (with respect to the usual sum and Cauchy product); if~$\K$ is a field,~$\laurent{\K}{z}$ plays the role of the field of fractions of the integral domain~$\K[[z]]$. If~$L$ is in~$\laurent{\K}{z}$, then its inverse can be written~$L^{-1} = z^{-k} (a_k + a_{k+1}z + a_{k+2}z^2 +\cdots)^{-1}$.

If~$\K$ is a field then the field of {\em Puiseux series} with coefficients in~$\K$ is defined as the set of Laurent series~$\sum_{n\ge k} a_n z^{n/N}$, where~$N\in \Z_{>0}$ is a positive integer,~$k\in \Z$, and each~$a_n$ belongs to~$\K$.~$\puiseux{\K}{z}$ sometimes denotes the field of Puiseux series (with respect to the usual sum and Cauchy product). If~$\K$ is algebraically closed and has characteristic 0, then the field of Puiseux series over~$\K$ is the algebraic closure of the field of Laurent series over~$\K$ (see Stanley~\cite[Vol.~2, \S6.1]{Stanley99} for a comprehensive approach). Laurent series and Riordan arrays are well investigated in~\cite{He11}, for example.

A \emph{generating function} exists as an element of~$\C[[z]]$, $\R[[z]]$ (or extended to a \emph{multivariate} generating functions, e.g. bivariate ones in~$\C[[z,u]]$), whatever the radius of convergence (the series can be divergent). Precisely, if~$(f_n)$ ($n\in \N$) is a sequence of real or complex numbers, the power 
series~$f(x) = \sum_{n\ge 0} f_n \omega_n x^n$ is then called the generating function (GF) of the sequence~$(f_n)$ with respect to a given reference sequence~$(\omega_n)$ of non-zero values (see e.g. \cite{Comtet74}).

\subsection{Application to combinatorial structures} \label{app:combistruct}
Exponential generating functions (EGFs) of the form
\[
G(z) = \frac{1}{\sqrt[\uproot{2} d]{1 - dz}}\ \quad (d\in \Z_{>0})\] enumerates several sorts of various combinatorial structures, e.g. varieties of increasing ordered rooted trees on~$d$ or~$d + 1$ vertices, weights of Dyck~$d$-paths, multiple factorials, etc. (see the examples below). Note that there exists a one-to-one correspondence between~$G(z)$ and the EGFs~$g_d(z)$ for all~$d\in \Z_{>0}$ defined by substituting~$z^d$ for~$z$ 
in~$G(z)$.

Consider for instance the~$(d+1)$-ary increasing trees, the~$d$-plane recursive trees and the~$d$-Stirling permutations defined in~\cite{GeSt78}. Recall that, for~$d\ge 1$, the degree-weight generating function of~$(d+1)$-ary increasing trees is given by~$\varphi(t) = (1+t)^{d+1}$, i.e.~$\varphi_0 = 1$. Consequently, the generating function~$G(x)$ and the numbers~$G_n$ of~$(d+1)$-ary trees of order~$n$ are obtained by
\[
G(z) = \frac{1}{(1 - dz)^{1/d}} - 1,\ \qquad G_n = n!\l[z^n\r]G(z) = \prod_{j=0}^{n-1} (jd + 1)\ \ %
\text{for}\ n\ge 1\ \ \text{and}\ \ G_0 = 0.\]
In addition,~$G_n = Q_n$, the number of~$d$-Stirling permutation, which helps have a combinatorial interpretation of Gessel's following theorem (see e.g.~\cite{BeFS92,FlSe09} and~\cite[p.~7]{JaKP11} for a detailed proof).

\m \no {\bf Theorem}\ {\rm (Gessel)}\ \emph{Let~$d\ge 1$. The family~$\cT_n(d+1)$ of~$(d+1)$-ary increasing trees of order~$n$ is in a natural bijection with~$d$-Stirling permutations:~$\cT_n(d+1)\simeq \cQ_n(d)$.}\par
For~$d = 1$, the well known bijection between~$1$-Stirling permutations (ordinary permutations) and binary increasing trees is recovered.

\m \no \textit{Example}\ For any integer~$n\ge 1$, consider the EGF
\[ 
G(z) = \frac{1}{\b(1 - dz\b)^{1/d}} - 1\ \quad \text{(with radius of convergence~$1/d$)}.\]

For~$d = 1$, the sequence of coefficients~$G_n = n!\l[z^n\r]\frac{1}{(1-z)}-1$ is~$n!$.

For~$d = 2$, the sequence~$G_n = n!\l[z^n\r]\frac{1}{\sqrt{1 -2z}}-1 = (1/2)_n\,2^n$,\ where, for any real number~$a$,~$(a)_n := a(a+1)\cdots (a+n-1)$ is the \emph{Pochhammer symbol} (or the \emph{rising factorial}~$a^{\ov{n}}$) of~$a$, starts as\ $1, 1, 3, 15, 105, 945, 10395, 135135, 2027025, 34459425,\ldots$, which appears as the sequence \emph{OEIS}~\href{http://oeis.org/A001147}{A001147} in~\cite{Sloane}. It is the double factorial of odd numbers:~$(2n-1)!! := 1\,3\,5\cdots (2n-1)$, which counts the number of increasing ordered rooted trees on~$n+1$ vertices, increasing ternary trees on~$n$ vertices or also the total weight of all Dyck~$n$-paths when each path is weighted with the product of the heights of the terminal points of its upsteps, etc.

For~$d = 3$, the sequence of coefficients~$G_n = n!\l[z^n\r](1 - 3z)^{-1/3}-1 = (1/3)_n\, 3^n$ starts as\ 
$1, 4, 28, 280, 3640, 58240, 1106560, 24344320, 608608000, 17041024000,\ldots$; that is the triple factorial~$(3n-2)!!! = 3^n\,(1/3)^{\ov{n}}$ with leading 1 added, which counts the number of increasing quaternary trees on~$n$ vertices: see \emph{OEIS}~\href{http://oeis.org/A007559}{A007559}.

Now, turning to the EGF~$(1-4z)^{-1/2}-1$, the coefficients~$n!\l[z^n\r] (1-4z)^{-1/2}-1 = (1/2)_n\, 4^n$ start as\ $1, 2, 12, 120, 1680, 30240, 665280, 17297280, 518918400,\ldots$; that is the quadruple factorial numbers~$(2n)!/n!$, which counts the binary rooted trees (with out-degree~$\le 2$) embedded in the plane with~$n$ labeled end nodes of degree 1: see \emph{OEIS}~\href{http://oeis.org/A001813}{A001813}. (The unlabeled version gives the Catalan numbers~\emph{OEIS}~\href{http://oeis.org/A000108}{A000108}.)

The coefficients of the EGF~$\l(1-2z^2\r)^{-1/2} - 1$ equal~$(1/2)_n\,2^n\, (2n)!/n!$, which  start as\ $1, 2, 36, 1800, 176400, 28576800, 6915585600, 2337467932800\ldots$; that is the number of functions~$f : \{1,2,\ldots ,2n\} \lra \{1,2,\ldots ,2n\}$ such that each element has either 0 or 2 preimages (i.e.~$\l|(f^{-1}(x)\r|\in \{0,2\}$ for all~$x$ in~$\{1,2,\ldots ,2n\}$) or binary mappings, the definition of which also involves planted trees and binary trees (see~\cite[p.~331]{FlSe09}): the unsigned \emph{OEIS} \href{http://oeis.org/A126934}{A126934}.

\end{document}